\documentclass{IEEEtran}
\usepackage[latin9]{inputenc}
\usepackage{float}
\usepackage{color}
\usepackage{amsmath}
\usepackage{amsthm}
\usepackage{amssymb}
\usepackage{graphicx}

\makeatletter

\floatstyle{ruled}
\newfloat{algorithm}{tbp}{loa}
\providecommand{\algorithmname}{Algorithm}
\floatname{algorithm}{\protect\algorithmname}

  \theoremstyle{definition}
  
  \theoremstyle{plain}
  \newtheorem{thm}{\protect\theoremname}
  \theoremstyle{plain}
  \newtheorem{lem}{\protect\lemmaname}
  \theoremstyle{definition}
  
  \theoremstyle{remark}
  \newtheorem{rem}{\protect\remarkname}
  \theoremstyle{plain}
  
 \newtheorem{ass}{Assumption}
  \theoremstyle{assumption}
\newtheorem{exa}{Example}
\newtheorem{cor}{Corollary}

\@ifundefined{date}{}{\date{}}
\usepackage{algorithm}
\usepackage{algorithmic}

\providecommand{\definitionname}{Definition}
\providecommand{\lemmaname}{Lemma}
\providecommand{\problemname}{Problem}
\providecommand{\propositionname}{Proposition}
\providecommand{\remarkname}{Remark}
\providecommand{\theoremname}{Theorem}

\makeatother

\begin{document}

\title{{Convergence and Accuracy Analysis for A Distributed Static State Estimator based on Gaussian Belief Propagation}}

%\thanks[footnoteinfo]{This paper was in part presented at the 54th IEEE Conference on Decision and Control.}
%\thanks{This work was in part supported by the National Natural Science Foundation of China (No. 61304038).}

%\thanks[footnoteinfo]{This paper has not been presented at any IFAC meeting. }
%\author[zheda]{Tianju~Sui}\ead{suitj@zju.edu.cn},    % Add the
%\author[newcastle]{Damian~Marelli}\ead{Damian.Marelli@newcastle.edu.au},               % e-mail address
%\author[zheda,newcastle]{Minyue~Fu\corauthref{cor}}
%
%\address[zheda]{Department of Control Science and Engineering, Zhejiang University, Hangzhou, 310013, China}  % Please supply
%\address[qinghua]{Department of Automation, Tsinghua University, Beijing, 100084, China.}             % full addresses
%\address[newcastle]{School of Electrical Engineering and Computer Science, The University of Newcastle, NSW 2308, Australia.}        % here.
\author{Dami\'{a}n Marelli, Tianju Sui$^\dag$, Minyue Fu,~\emph{Fellow IEEE} and Ximing Sun
  \thanks{Dami\'{a}n Marelli is with the School of Automation, Guangdong University of Technology, China, and with the French Argentine International Center for Information and Systems Sciences, National Scientific and Technical Research Council, Argentina (email: Damian.Marelli@newcastle.edu.au). Tianju Sui and Ximing Sun are with the School of Control Science and Engineering, Dalian University of Technology, China (email:{suitj@mail.dlut.edu.cn;~sunxm@dlut.edu.cn}). Minyue Fu is with School of EE\&CS, University of Newcastle, Australia. He also holds an Bai-ren Professorship at the School of Automation, Guangdong University of Technology, China (email: minyue.fu@newcastle.edu.au).}
  \thanks{$\dagger$Corresponding author.}
  \thanks{This work was supported by the National Natural Science Foundation of China (Grant Nos. 61803068, 61633014, 61803101 and U1701264), the Fundamental Research Funds for the Central Universities (DUT19RC(3)027), and the Argentinean Agency for Scientific and Technological Promotion (PICT- 201-0985).}}

\maketitle

%\received{July 8, 2008.}
%
%\acks{Minyue Fu is the corresponding author.}

\begin{abstract}
  This paper focuses on the distributed static estimation problem and a Belief Propagation (BP) based estimation algorithm is proposed. We provide a complete analysis for convergence and accuracy of it. More precisely, we offer conditions under which the proposed distributed estimator is guaranteed to converge and we give concrete characterizations of its accuracy. Our results not only give a new algorithm with good performance but also provide a useful analysis framework to learn the properties of a distributed algorithm. It yields better theoretical understanding of the static distributed state estimator and may generate more applications in the future.
\end{abstract}


%\corauth[cor]{Corresponding author.}\ead{minyue.fu@newcastle.edu.au}  % (ead) as shown

\section{Introduction \label{sec:Introduction}}
Large-scale systems, such as the power grid, sensor monitoring network and the telecommunication system,
are receiving increasing attention from researchers in different fields. The data size among the system is rapidly increasing and the classical centralized control/estimation method is not suitable for the "big data" case. Thus, the distributed approach is urgently required for each node in the network can locally work. We focus on the distributed state estimation problem in this paper and a brief review on this field is given.

Existing distributed estimation methods can be classified in two ways. The first classification is done according to the dynamics of system model. In the
static system, the state of each node is static and estimated using its own measurements and information from its neighbour nodes. In the dynamic system, the state of each node is time-varying and may be correlated with the states of other nodes. The second classification is done according to the scale of system. In the small-scale system, each node focus on the same global state and they work in parallel to estimate it. In the large-scale system, each node focus on the its own state of interest and the shared information is correlated with the states of different nodes. It is impossible for a node in the large-scale system to estimate the whole global state for that it is too long. For example, in a traffic system for a city, due to the spatial correlations of the traffic flows in different nodes, neighboring traffic information is certainly useful in predicting the traffic conditions at
each node. Taking the traffic conditions as the states, the global state consist of every nodes's conditions is too huge and also unnecessary for a node.

Next, we briefly summarize the existing works on these four classes of systems.

In the static small-scale system, there are mainly two distributed approaches. The first approach is the consensus based algorithm and the representative works are given in \cite{Garin2010},\cite{olfati2004consensus},\cite{Schizas2007},\cite{Schizas2009}. They ran average consensus on the information vector and information matrix of each node and the final state estimate of each node asymptotically converge to the optimal one, i.e., Weighted Least Square (WLS) one. The second approach is the iterative projection based algorithm and it is proposed in \cite{pasqualetti2012distributed}. This estimation technique belongs to the family of Kaczmarz methods~\cite{Tanabe1971},\cite{Censor1981} for the solution of a linear system of equations. The iterative projection based algorithm only requires finite iterations for the convergence towards the exact WLS solution of a system.

In the static large-scale system, the goal is that the composite estimate of the whole system, consisting of all local estimates,
will approach the optimal estimate obtained using all the measurements and
a centralized estimation method~\cite{Conejo2007},\cite{Gomez2011},\cite{damian2015wls}. The technical difficulty for a large-scale system is higher than that for a small-scale one and the small-scale distributed estimation problem can be viewed as a special case of the large-scale one, i.e., the distributed algorithms for large-scale system is always effective for the small-scale system. It is worth mentioning that \cite{damian2015wls} gave a
novel way based on the Richardson method for solving linear equation to approach the WLS estimate of each node's state, and it is advanced among the existing works.

In the dynamic small-scale system, the consensus based algorithm is also popular. In \cite{matei2012}, the authors chose to run one time of consensus on the independent estimate (each node use its own measurement independently) each sampling period. And the authors in \cite{Cattivelli2010} ran consensus on the local estimate (each sensor use its own and neighbors' measurements). While, the stability conditions in \cite{matei2012} and \cite{Cattivelli2010} are stronger than that of the global optimal estimator. The work in \cite{Acikmese2014} studied how many times of consensus on the independent estimate between each sampling period is sufficient to guarantee the estimation stability under the least observability condition (i.e., the stability condition for the global optimal estimator). Besides the consensus on the estimate, \cite{Battistelli2014} found that, by
running consensus on the information matrix and information vector once each sampling period, the least observability condition is sufficient for the estimation stability.

In the dynamic large-scale system, \cite{zhou2013coordinated},\cite{farina2010moving},\cite{khan2008distributing},\cite{zhou2015controllability} have done an
elaborate system analysis and proposed some information
passing/processing method to get a stable estimate. While,
\cite{haber2013moving} considered a class of system with banded dynamic
system transition matrix and found that the contribution from
faraway nodes decreases with the increase of distance. It
showed that the moving horizon estimation could offer a
approximated optimal state estimate in that case.

We choose to study the distributed state estimation problem for static large-scale system in this paper. It has many application areas and we choose two of them as examples.

The power system requires to estimate the voltages and phases of the power supply at each
sub-system, consisting of a number of buses or a substation, using measurements (for example, a phasor measurement
unit (PMU) or a supervisory control and data acquisition (SCADA)
system). This concern was first recognized and addressed in \cite{Schweppe1970i},\cite{Schweppe1970ii},\cite{Schweppe1970iii} by introducing the idea of static state estimation into power systems. Interactions of sub-systems are reflected by the fact that local measurements
available at each sub-system typically involve neighboring
sub-systems. For example, a current measurement at a conjunction
depends on the voltage difference of two neighboring buses.
In a smart grid setting, each sub-system is only interested in its
local state, i.e., its own voltages and phases, using local measurements
and information acquired from neighboring sub-systems via
neighborhood communication~\cite{Taixin2013acc}. Thus, distributed methods for local state estimation in static large-scale system are naturally called for.

{The sensor network localization problem involves estimating the locations
of all sensors using relative measurements (e.g., relative distances
or relative positions) between sensors~\cite{Pal2010}. It is also unnecessary
for each sensor to localize other nodes. A distributed method
is preferred, where each node aims to estimate its own location
using local measurements and neighborhood communication~\cite{khan2009local},\cite{Moore2004},\cite{ravazzi2018distributed}. In \cite{battilotti2018cooperative}, the authors focused on the distributed linear system with relative non-linear measurements and it is very suitable for the formation control.}

To solve the distributed static state estimation problem for large-scale system, motivated by the idea from Belief Propagation, a novel distributed estimator is given. \emph{Pearl's Belief Propagation}, or \emph{Belief Propagation} (BP)
for short, is a well-celebrated algorithm for solving the above problem.
Originally proposed by Pearl~\cite{pearl1988probabilistic} in 1982,
BP (also known as {\em sum-product message passing}), is a {\em
message passing} algorithm for computing marginal PDFs on Bayesian
networks (directed and acyclic graphs) and Markov random fields (undirected
and cyclic graphs). Since its introduction, BP has been widely accepted
as a powerful distributed algorithm in many scientific and engineering
fields, including artificial intelligence, information theory, applied
mathematics, signal processing and control systems~\cite{Braunstein}.
Renowned applications of BP include low-density parity-check codes
and turbo codes for digital communications~\cite{mceliece1998turbo,berrou1996near,Fu_turbo},
free energy approximation for statistical learning~\cite{Braunstein},
satisfiability for mathematical logic~\cite{Braunstein}, combinatorial
optimization~\cite{braunstein2004survey} and computer vision~\cite{sun2003stereo,tanaka2002statistical}.
BP also finds important applications in the area of state estimation.
It is interesting to note that the famous Kalman filtering algorithm
for state estimation of dynamic systems is known to be an example
of BP~\cite{Kschischang}.

For acycilc graphs (i.e., graphs without cycles), it is known that
BP converges in a finite number of iterations, and the correct marginals
will be produced~\cite{pearl1988probabilistic,weiss2001correctness}.
For cyclic graphs (i.e., graphs with cycles), BP is not guaranteed
to converge in general, and even if it does, it may not calculate
the correct marginals. Nevertheless, the wonderful and mysterious
feature of BP is that for most applications, BP delivers amazingly
good approximations for the marginals, despite the existence of cycles~\cite{frey1998graphical,murphy1999loopy,freeman2000learning}.
Turbo decoding is perhaps most successful example of such a BP application,
as it delivers near-Shannon-capacity performances, despite the fact
that the underlying graph can be very loopy. This success has been
claimed as ``the most exciting and potentially important development
in coding theory in many years''~\cite{mceliece1995turbo}.

Gaussian BP is the BP algorithm specialized to Gaussian distributions.
The algorithm computes iteratively the mean and variance (or covariance)
of each marginal. Gaussian BP has been successfully applied in low
complexity detection and estimation problems arising in communication
systems~\cite{montanari2005belief,guo2011based,guo2008lmmse}, fast
solver for large sparse linear systems~\cite{el2012efficient,shental2008gaussian},
sparse Bayesian learning~\cite{tan2010computationally}, estimation
in Gaussian graphical model~\cite{chandrasekaran2008estimation}, distributed
beam forming~\cite{ng2008distributed}, inter-cell interference mitigation~\cite{lehmann2012iterative},
distributed synchronization and localization in wireless sensor networks~\cite{leng2011distributed,ahmad2012factor,leng2011cooperative},
distributed energy efficient self-deployment in mobile sensor networks~\cite{song2014distributed},
distributed rate control in Ad Hoc networks~\cite{zhang2010fast},
distributed network utility maximization~\cite{dolev2009distributed},
and large-scale sparse Bayesian learning~\cite{seeger2010variational}.

BP's excellent performances have inspired {\em many} researchers
over the last 20 years to study its theoretical properties. The most fundamental
question is that, for a cyclic network graph, under what conditions
will BP iterations converge? For a general cyclic graph,~\cite{tatikonda2002loopy,tatikonda2003convergence,ihler2004message,mooij2007sufficient,ihler2005loopy}
studied the convergence condition for BP. However, these references only gave
partial answers. Several conditions ensuring the convergence of the marginals under a designated initialization
set have been proposed~\cite{weiss2001correctness,malioutov2006walk,su2014convergence,su2015convergence}.
But several major drawbacks exist. Firstly, the convergence conditions
are too difficult to check. For examples, \cite{su2014convergence} and \cite{Moallemi2010}
require to check a very complex condition for the convergence. The convergence condition for the
mean in \cite{su2015convergence} requires the evaluation of the spectral
radius of an infinite dimensional matrix, which is impossible in practice.
Secondly, the convergence conditions are too strict. For examples, \cite{weiss2001correctness} shows that, under the "strictly diagonally dominant" requirement for measurement parameters, the state estimate of Gaussian Belief Propagation converges. \cite{Johnson2006} extends the work of \cite{weiss2001correctness} by relaxing the assumption to generalised diagonal dominance. {While, both diagonal dominance requirements for convergence are too strict. Thirdly, the convergence rate are not quantified in tidy form. For example, the convergence rate in \cite{Moallemi2010} is not explicit. Finally, most of these convergence analysis are done only for scalar sub-systems except \cite{du2017convergence}(i.e., the sub-state of each node is a scalar\footnote{Since the state components for each node are not independent for a
vector system, results for scalar systems may not applicable to vector
systems.})}

{In this paper, based on the Gaussian BP, a new distributed state estimator is proposed for static large-scale system.} We will study the convergence and accuracy properties of the distributed estimator under the general setting and a vector system. The goal of distributed state estimator is to devise a distributed iterative algorithm so that each node will compute a good estimate of its own state using its own measurements and information exchange with its neighbours.

{Viewing the iterations in our proposed algorithm as a dynamic process, we will provide conditions under which
its stability (i.e., the convergence of the distributed state estimator) is guaranteed. Comparing with the convergence condition in \cite{du2017convergence}, we focus on the vector system and also determine the convergence speeds of the state estimate and covariance. We will show that
the estimation error generated by our algorithm can be quantified. More specifically,
our estimation error and covariance formulas clearly explain the impact
of the so-called {\em cycle-free depth} of each node to the accuracy.}

The significance of our work lies in both theoretical contributions and potential applications. Firstly, convergence and accuracy analysis for the static distributed state estimator is theoretically meaningful. Secondly, the proposed distributed estimation algorithm has very good performance on a static large-scale system with large {\em cycle-free depth}. Due to the fact that the static estimation is a fundamental technique with vast applications,
our distributed static estimation is naturally needed for large-scale networked systems when centralized solutions are not possible.

The rest of this paper is organized as follows. In Section~\ref{sec:formulation},
the problem formulation and distributed state estimator are proposed.
Section~\ref{sec:ConvergenceQ} studies the convergence of the information
matrices. Section~\ref{sec:ConvergenceX} investigates the convergence
of the state estimates. The accuracy of the information matrices and state estimates
are analyzed in Section~\ref{sec:AccuracyQ} and \ref{sec:AccuracyX},
respectively. In Section~\ref{sec:simulation}, simulations are given
to illustrate our results. Concluding remarks are stated in Section~\ref{sec:conclusion}. Some proofs of complementary results are left in the Appendix.

\section{Preliminary for Belief Propagation}\label{sec:pre}
{The BP algorithm concerns with a system represented by a {\em
bipartite graph} with $I$ \emph{variable nodes} and $V$ \emph{factor
nodes}, as depicted in the Fig.~\ref{factor_graph}. Each variable
node $i$ is associated with a random vector $x_{i}\in\mathbf{R}^{n_{i}}$
and each factor node $v$ is connected to a subset of variable nodes,
$\mathcal{F}_{v}\subset\{1,2,\ldots,I\}$. Denoting the joint (or
global) variable by $X=\{x_{i}:i=1,2,\ldots,I\}$, it is assumed that
its joint probability density function (PDF) $f(X)$ can be expressed in a factor form:
\[
f(X)=\prod_{v=1}^{V}f_{v}\left(X_{v}\right),
\]
where $X_{v}=\{x_{i}:i\in\mathcal{F}_{v}\},v=1,2,\ldots,V$. Each
$f_{v}(X_{v})$ represents a piece of partial ``knowledge'' about
$X$.}
\begin{figure}[ht]
\begin{centering}
\includegraphics[width=8cm]{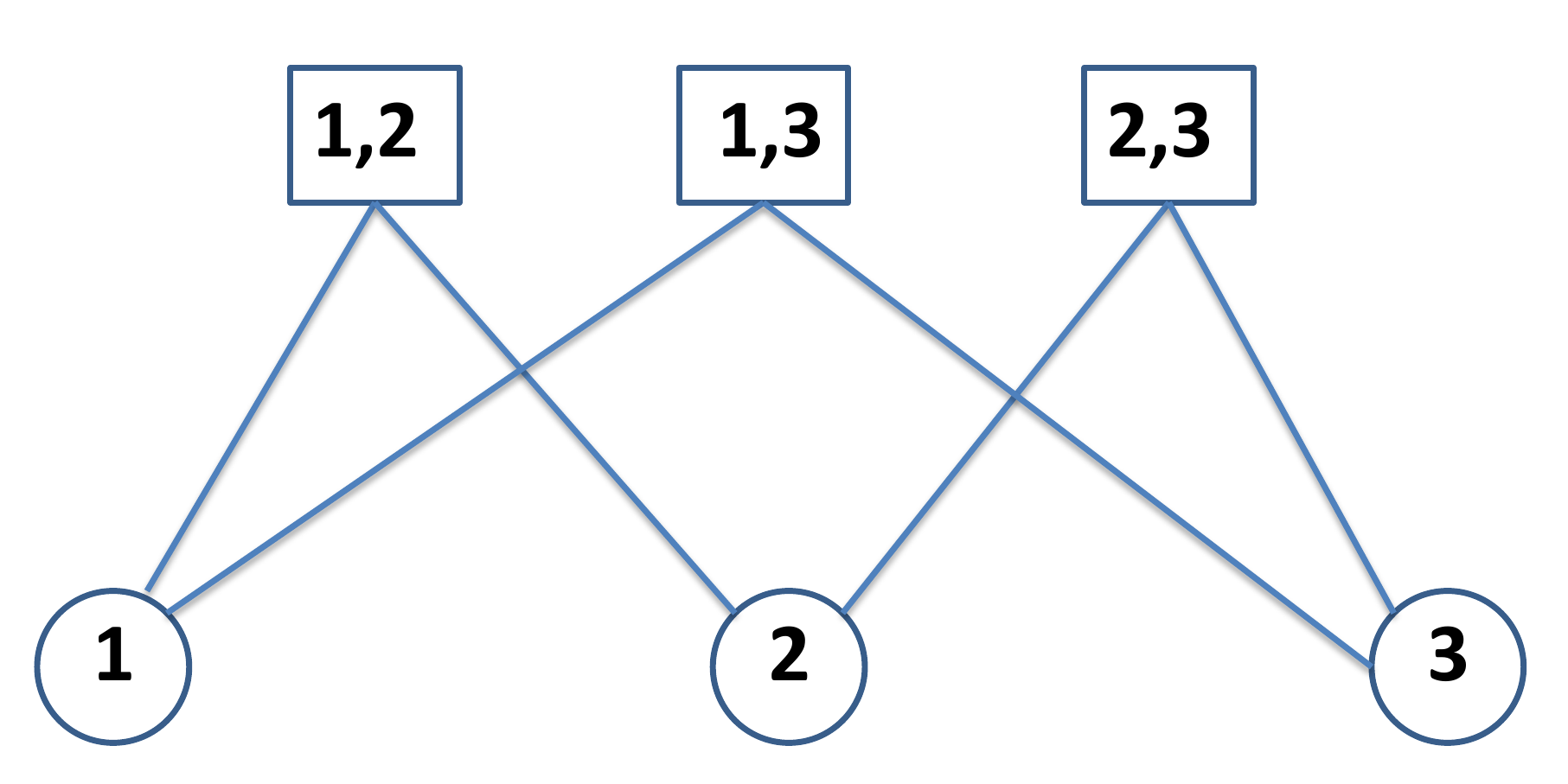}
\par\end{centering}

\caption{Bipartite graph: circles = variable nodes; squares = factor nodes}
\label{factor_graph}
\end{figure}

The goal of BP is to compute, at each node~$i$, the marginal $g_{i}(x_{i})$
of $f(X)$, which is defined by
\begin{equation}
g_{i}(x_{i})=\int f(X)d(X\setminus x_{i}),\label{marginal}
\end{equation}
where $X\setminus x_{i}$ is the set obtained from $X$ by removing
$x_{i}$. The algorithm resorts to iterative computation and local
communication between connected variable nodes and factor nodes. More
specifically, the algorithm starts by each factor node $v$ sending
to each variable node $i\in\mathcal{F}_{v}$ the following marginal
PDF (called {\em message})
\begin{equation}
m_{v\rightarrow i}^{(0)}(x_{i})=\int f_{v}(X_{v})d(X_{v}\setminus x_{i}).\label{eq:BP0}
\end{equation}
{Then, at each iteration $k=1,2,\ldots$, each variable node $i$ sends
to every connected factor node $v$ the following message:}
\begin{equation}
m_{i\rightarrow v}^{(k)}(x_{i})=\prod_{w\in\mathcal{N}_{i}\setminus v}m_{w\rightarrow i}^{(k)}(x_{i}),\label{eq:BP1}
\end{equation}
where $\mathcal{N}_{i}$ is the set of factor nodes connected to variable node
$i$. Similarly, each factor node $v$ sends to every connected variable
node $i$ the following message:
\begin{equation}
m_{v\rightarrow i}^{(k)}(x_{i})=\int f_{v}(X_{v})\prod_{j\in\mathcal{F}_{v}\setminus i}m_{j\rightarrow v}^{(k-1)}(x_{j})d(X_{v}\setminus x_{i}).\label{eq:BP2}
\end{equation}
The desired marginal at node $i$ and iteration $N$ is estimated:
\begin{equation}
g_{i}^{(N)}(x_{i})=\prod_{w\in\mathcal{N}_{i}}m_{w\rightarrow i}^{(N)}(x_{i}),\label{eq:GP-mar}
\end{equation}
modulo a constant scalar to make its integral equal 1.

It is clear that this is a {\em fully distributed algorithm} because
only local information gets exchanged and used without the need for
any global information. Note that BP is also known as the {\em sum-product}
algorithm because (\ref{eq:BP0})-(\ref{eq:BP2}) use sums (i.e.,
integrals) and products.

\section{Problem Formulation}\label{sec:formulation}

Consider a system with $I$ unknown sub-states\footnote{{Since we focus on the static estimation problem, the sub-states $x_{1},x_{2},\ldots,x_{I}$ are unknown time-invariant vectors.}} $x_{1},x_{2},\ldots,x_{I}$ and each of them corresponds to a sensing node. Following the notations in BP, we call the sensing nodes {\em variable nodes} and $x_{i}\in \mathbb{R}^{n_i}$ is a vector. Associated with the system are two kinds of measurements (also vectors), the so-called {\em self measurement}
for node $i$,
\begin{equation}
z_{i}=C_{i}x_{i}+v_{i},\label{eq:cycsys1}
\end{equation}
and (pair-wise) {\em Joint measurement} between nodes $i$ and
$j$,
\begin{equation}
z_{i,j}=C_{i,j}x_{i}+C_{j,i}x_{j}+v_{i,j}.\label{eq:cycsys2}
\end{equation}
In the above, the matrices $C_{i},C_{i,j}$ and $C_{j,i}$ are known;
$v_{i}$ and $v_{i,j}$ are independent Gaussian measurement noises with known covariances $R_{i}>0$ and $R_{i,j}>0$,
respectively. Note that: 1) the factor node $(i,j)$ is unordered,
i.e., $(i,j)=(j,i)$; 2) $z_{i,j}=z_{j,i}$ and $v_{i,j}=v_{j,i}$;
3) It is not necessary for all variable nodes to have self measurements
or all variable node pairs to have joint measurements. In fact, joint
measurements are typically sparse for large graphs.

The problem of distributed WLS estimation is to compute the maximum
likelihood (ML) estimate for each $x_{i}$ and the corresponding estimation
error covariance using a fully distributed algorithm. It is clear
that the likelihood functions given by the self and joint measurements
are, respectively,
\begin{align}
f_{i}(x_{i}) & =p\left(z_{i}|x_{i}\right)\nonumber \\
 & \sim\mathcal{N}(z_{i}-C_{i}x_{i},R_{i}),\label{Gaussian_fi}\\
f_{i,j}(x_{i},x_{j}) & =p\left(z_{i,j}|x_{i},x_{j}\right)\nonumber \\
 & \sim\mathcal{N}(z_{i,j}-C_{i,j}x_{i}-C_{j,i}x_{j},R_{ij}),\label{Gaussian_fij}
\end{align}
where $\mathcal{N}(\mu,\Sigma)$ stands for a Gaussian PDF with mean
$\mu$ and covariance $\Sigma$. {It is noted that, in our setting, $f_{i}(x_{i})$ in \eqref{Gaussian_fi} corresponds to the variable node and $f_{i,j}(x_{i},x_{j})$ in \eqref{Gaussian_fij} corresponds to the factor node.}

The joint likelihood function for
$X=\{x_{i}:i=1,2,\ldots,I\}$ becomes
\begin{equation}
f(X)=\prod_{i}f_{i}(x_{i})\prod_{(i,j)}f_{i,j}(x_{i},x_{j}).\label{Gaussian_f}
\end{equation}
Therefore, the maximum likelihood function for each $x_{i}$ is given
by
\begin{equation}
g_{i}(x_{i})=\int f(X)d(X\setminus x_{i}),\label{marginal}
\end{equation}
which is exactly the task of static distributed state estimator.

Denote the distributed state estimation and covariance at iteration $N$ by $\hat{x}_{i}(N)$ and $\Sigma_{i}(N)$, respectively. They are calculated by Algorithm~\ref{BP}. Moreover, denote
\[
\Omega_{i,j}=C_{i}^{T}R_{i}^{-1}C_{i}+\sum_{w\in\mathcal{N}_{i}\backslash j}C_{i,w}^{T}R_{i,w}^{-1}C_{i,w}.
\]
{Notice that $\Omega_{i,j}>0$ for every $(i,j)$, is necessary for Algorithm~\ref{BP}.}

Defining
\begin{equation}
\alpha_{i}(N)=Q_{i}(N)\hat{x}_{i}(N);\ \ Q_{i}(N)=\Sigma_{i}^{-1}(N),\label{alpha}
\end{equation}
which we call {\em information vector} and {\em information matrix},
or {\em information parameters} collectively. It is easy to verify that $\Omega_{i,j}=Q_{i\rightarrow i,j}(1)$ for all $i=1,2,\ldots,I$.

{Throughout this paper, our discussions are always under the following assumption.}

\begin{ass}\label{asu:1} For all $i=1,2,\ldots,I$ and $j\in\mathcal{N}_{i}$,
\begin{equation}
\Omega_{i,j}=C_{i}^{T}R_{i}^{-1}C_{i}+\sum_{w\in\mathcal{N}_{i}\backslash j}C_{i,w}^{T}R_{i,w}^{-1}C_{i,w}>C_{i,j}^{T}R_{i.j}^{-1}C_{i,j}.  \label{assump:01}
\end{equation}
\end{ass}
\begin{rem}
\label{rem:assumption0}
{Roughly speaking, \eqref{assump:01} means that, for each node $i$,
the information contribution from any single neighbouring node $j$
(i.e., $C_{i,j}^{T}R_{i,j}^{-1}C_{i,j}$) is strictly smaller than
the sum of that from node $i$ (i.e., $C_{i}^{T}R_{i}^{-1}C_{i}$) and all other
neighbouring nodes $w\in \mathcal{N}_{i}\backslash j$ (i.e., $C_{i,w}^{T}R_{i,w}^{-1}C_{i,w}$).}
In particular,  Assumption~\ref{asu:1} implies
that $\Omega_{i,j}>0$ for all $(i,j)$. It also implies that, for
every leaf node~\footnote{A variable node is called a {\em leaf
node} if it is connected  by only one edge.} $i$, $C_{i}^{T}R_{i}^{-1}C_{i}>0$ (or equivalently, $C_{i}$ has
full column rank), due to the fact that the sum term in $\Omega_{i,j}$
is void in this case.
\end{rem}

\begin{rem}
\label{rem:assumption} Due to the strict inequality above, it is
clear that Assumption~\ref{asu:1} is equivalent to the existence
of some constant $0<\eta<1$ such that
\begin{equation}
\eta\Omega_{i,j}\ge C_{i,j}^{T}R_{i,j}^{-1}C_{i,j}\label{eq:as1}
\end{equation}
for all $j\in\mathcal{N}_{i}$. We will use this property in the sequel.
\end{rem}

\begin{algorithm}[!]
\protect\protect\protect\protect\protect\protect\protect\caption{{A BP-based Distributed Static State Estimator}}
\label{BP} \begin{algorithmic}

\STATE 1) \textbf{Initialization:} At time $k=0$, factor node $(i,j)$
sends to each connected variable node $i$:
\begin{align}
\alpha_{i,j\rightarrow i}(0) & =C_{i,j}^{T}R_{i,j}^{-1}z_{i,j},\nonumber\\
Q_{i,j\rightarrow i}(0) & =C_{i,j}^{T}R_{i,j}^{-1}C_{i,j}.\label{newinitial}
\end{align}

\STATE 2) \textbf{Main loop:} At time $k=1,2,\cdots$, do:

2.1) Each variable node $i$ computes
\begin{align}
\alpha_{i}(k) & =C_{i}^{T}R_{i}^{-1}z_{i}+\sum_{j\in\mathcal{N}_{i}}\alpha_{i,j\rightarrow i}(k-1),\nonumber\\
Q_{i}(k) & =C_{i}^{T}R_{i}^{-1}C_{i}+\sum_{j\in\mathcal{N}_{i}}Q_{i,j\rightarrow i}(k-1),\label{eq:BP-Q}
\end{align}
and (if required at this iteration)
\begin{align}
\hat{x}_{i}(k) & =Q_{i}^{-1}(k)\alpha_{i}(k),\nonumber\\
\Sigma_{i}(k) & =Q_{i}^{-1}(k).\label{eq:BP-x-hat}
\end{align}
2.2) {Each variable node $i$ sends to each adjacent variable node $(i, j)$ with $j\in\mathcal{N}_{i}$}:
\begin{align}
\alpha_{i\rightarrow i,j}(k) & =\alpha_{i}(k)-\alpha_{i,j\rightarrow i}(k-1),\nonumber \\
Q_{i\rightarrow i,j}(k) & =Q_{i}(k)-Q_{i,j\rightarrow i}(k-1),\label{eq:BP-Q1}
\end{align}

2.3) {Each factor node $(i, j)$ sends to variable node $j$}:
\begin{align}
\alpha_{i,j\rightarrow j}(k) & =C_{j,i}^{T}R_{i,j\rightarrow j}^{-1}(k)z_{i,j\rightarrow j}(k),\nonumber\\
Q_{i,j\rightarrow j}(k) & =C_{j,i}^{T}R_{i,j\rightarrow j}^{-1}(k)C_{j,i},\label{eq:BP-Q2}
\end{align}
where
\begin{align}
z_{i,j\rightarrow j}(k) & =z_{i,j}-C_{i,j}Q_{i\rightarrow i,j}^{-1}(k)\alpha_{i\rightarrow i,j}(k),\nonumber\\
R_{i,j\rightarrow j}(k) & =R_{i,j}+C_{i,j}Q_{i\rightarrow i,j}^{-1}(k)C_{i,j}^{T}.\label{eq:BP-R}
\end{align}
\end{algorithmic}
\end{algorithm}
{Since variables are gaussian, messages containing innovation vector $\alpha_{i\rightarrow i,j},\alpha_{i,j\rightarrow i}$ and information matrix $Q_{i\rightarrow i,j},Q_{i,j\rightarrow i}$ are equivalent to messages containing the distribution itself. The Algorithm \ref{BP} is an instance of Gaussian BP.}

To a given set of variable and factor nodes, we associate an undirected
graph, called the \emph{canonical graph}, which we denote by $\mathcal{G}$.
This graph has a node associated with each variable node $i=1,\ldots,I$,
and an edge between nodes $i$ and $j$, if there exists joint measurement $z_{i,j}$, i.e., for all $j\in\mathcal{N}_{i}$. Moreover, the edge $(i,j)$ is unordered and we also call it factor node $(i,j)$.

It is well known that the Algorithm~\ref{BP} converges to the correct
marginals in a finite number of iterations when $\mathcal{G}$ is
acyclic~\cite{Taixin2013acc}. In fact, the required number
of iterations equals to the diameter of the graph, i.e., the maximum distance of any pair of variable nodes, where the distance of two nodes is the number of edges of the shortest path between them. The fundamental
challenge in this paper is to understand how the algorithm performs
for cyclic graphs. As mentioned in Section~\ref{sec:Introduction},
the goal of this paper is of twofold: First, we want to provide conditions
to guarantee the convergence of Algorithm~\ref{BP} when the induced bipartite
graph is cyclic. Secondly, when convergence occurs, we want to quantify
the accuracy of the distributed state estimate, i.e., the difference between our state estimate and the true (or global) maximum likelihood estimate in \eqref{marginal}.

%In the rest of paper, without loss of generality, we concentrate our%study on the convergence and accuracy of an arbitrary variable node,%labelled node~1.

\section{Convergence Analysis for Information Matrices}\label{sec:ConvergenceQ}

In this section, we provide our first key result which shows that
the information matrices $Q_{i}(k)$ always converge exponentially
to a positive definite matrix, under Assumption 1. In addition, the rate of convergence
is also characterized.

Firstly, some preliminary lemmas are required.
\begin{lem}
\label{lem:01} For any $k\in\mathbb{N}$, $1\le i\le I$ and $j\in\mathcal{N}_{i}$,
\begin{align}
Q_{i\rightarrow i,j}(k+1) & \le Q_{i\rightarrow i,j}(k);\nonumber \\
Q_{i,j\rightarrow j}(k+1) & \le Q_{i,j\rightarrow j}(k);\nonumber \\
R_{i,j\rightarrow j}(k+1) & \ge R_{i,j\rightarrow j}(k).\label{mono}
\end{align}
In particular, $Q_{i\rightarrow i,j}(k)\le\Omega_{i,j}$ for all $k\ge1$.
\end{lem}
\begin{proof}
See Appendix~\ref{app:lem:01}.
\end{proof}

\begin{lem}
\label{lem:02} Under Assumption~\ref{asu:1}, for every $1\le i\le I$ and $j\in\mathcal{N}_{i}$, we have
\begin{align*}
Q_{i\rightarrow i,j}(\infty) & =\lim_{k\rightarrow\infty}Q_{i\rightarrow i,j}(k)>0;\\
R_{i,j\rightarrow j}(\infty) & =\lim_{k\rightarrow\infty}R_{i,j\rightarrow j}(k)<\infty.
\end{align*}
\end{lem}
\begin{proof}
See Appendix~\ref{app:lem:02}.
\end{proof}

Next, we give the main result on convergence. Define
\begin{align*}
\Delta Q_{i\rightarrow i,j}(k) & =Q_{i\rightarrow i,j}^{-1/2}(\infty)Q_{i\rightarrow i,j}(k)Q_{i\rightarrow i,j}^{-1/2}(\infty)-I;\\
\Delta Q_{i,j\rightarrow j}(k) & =Q_{i,j\rightarrow j}^{-1/2}(\infty)Q_{i,j\rightarrow j}(k)Q_{i,j\rightarrow j}^{-1/2}(\infty)-I;\\
\Delta R_{i,j\rightarrow j}(k) & =R_{i,j\rightarrow j}^{-1/2}(\infty)R_{i,j\rightarrow j}(k)R_{i,j\rightarrow j}^{-1/2}(\infty)-I.
\end{align*}
Also, let constants $\rho>0$ and $\alpha>0$ be defined as follows:
\begin{align}
\rho & =\max_{i,j}\Vert R_{i,j\rightarrow j}^{-1/2}(\infty)C_{i,j}Q_{i\rightarrow i,j}^{-1}(\infty)C_{i,j}^{T}R_{i,j\rightarrow j}^{-1/2}(\infty)\Vert,\label{rho}\\
\alpha & =\max_{i,j}\Vert Q_{i\rightarrow i,j}^{-1/2}(\infty)\Omega_{i,j}Q_{i\rightarrow i,j}^{-1/2}(\infty)-I\Vert.\label{alpha0}
\end{align}
Note that $\rho<1$ follows from $R_{i,j}>0$ and
\[
R_{i,j\rightarrow j}(\infty)=R_{i,j}+C_{i,j}Q_{i\rightarrow i,j}^{-1}(\infty)C_{i,j}^{T}.
\]

\begin{lem}
\label{lem:03} Under Assumption~\ref{asu:1}, for every node $i$, its neighbor node $j$ and all
$k\in\mathbb{N}$, we have
\begin{equation}
0\le\Delta Q_{i\rightarrow i,j}(k)\le\alpha\rho^{k-1}I.\label{eq:01}
\end{equation}
\end{lem}
\begin{proof}
See Appendix~\ref{app:lem:03}.
\end{proof}

Since ultimately we are only interested in the information matrices
$Q_{i}(k)$, we get the following result from Lemma~\ref{lem:03}.
\begin{thm}\label{thm:info}
Under Assumption~\ref{asu:1}, it holds that $Q_{i}(k)\rightarrow Q_{i}(\infty)>0$ as $k\rightarrow\infty$, for every node $i$ of $\mathcal{G}$. Moreover, by defining
\[
\Delta Q_{i}(k)=Q_{i}^{-1/2}(\infty)Q_{i}(k)Q_{i}^{-1/2}(\infty)-I,
\]
it holds, for every node $i$ of $\mathcal{G}$ and all $k\in\mathbb{N}$,
that
\begin{equation}
0\le\Delta Q_{i}(k)\le\alpha\rho^{k-1}I,\label{eq:Qi}
\end{equation}
where $\rho$ and $\alpha$ are defined in \eqref{rho} and \eqref{alpha0}, respectively. As analyzed before, we have $\rho<1$.
\end{thm}

\begin{proof}
From Lemma~\ref{lem:03}, for every node $j$ connected to
node $i$ and all $k\in\mathbb{N}$, we get
\begin{equation}
0\le Q_{i\rightarrow i,j}(k)-Q_{i\rightarrow i,j}(\infty)\le\alpha\rho^{k-1}Q_{i\rightarrow i,j}(\infty).\label{tempxxxx}
\end{equation}
In the proof of Lemma~\ref{lem:03}, we have (\ref{R-inverse}) which
says that
\[
0\le R_{j,i\rightarrow i}^{-1}(k-1)-R_{j,i\rightarrow i}^{-1}(\infty)\le\alpha\rho^{k-1}R_{j,i\rightarrow i}^{-1}(\infty).
\]
The left hand side above is also non-negative definite because of
the monotonic non-decreasing property of $R_{j,i\rightarrow i}(k)$.
It follows from \eqref{eq:BP-Q2} that
\begin{align*}
0 & \le Q_{i,j\rightarrow i}(k-1)-Q_{i,j\rightarrow i}(\infty)\\
 & =C_{i,j}^{T}(R_{j,i\rightarrow i}^{-1}(k-1)-R_{j,i\rightarrow i}^{-1}(\infty))C_{i,j}\\
 & \le\alpha\rho^{k-1}C_{i,j}^{T}R_{j,i\rightarrow i}^{-1}(\infty)C_{i,j}\\
 & =\alpha\rho^{k-1}Q_{i,j\rightarrow i}(\infty).
\end{align*}
Combining the above and \eqref{tempxxxx}, using \eqref{eq:BP-Q2}
yield
\begin{align*}
 & Q_{i}(k)-Q_{i}(\infty)\\
= & Q_{i\rightarrow i,j}(k)+Q_{i,j\rightarrow i}(k-1)-Q_{i\rightarrow i,j}(\infty)-Q_{i,j\rightarrow i}(\infty)\\
= & Q_{i\rightarrow i,j}(k)-Q_{i\rightarrow i,j}(\infty)+Q_{i,j\rightarrow i}(k-1)-Q_{i,j\rightarrow i}(\infty)\\
\le & \alpha\rho^{k-1}(Q_{i\rightarrow i,j}(\infty)+Q_{i,j\rightarrow i}(\infty))\\
= & \alpha\rho^{k-1}Q_{i}(\infty).
\end{align*}
The fact that $Q_{i}(k)-Q_{i}(\infty)\ge0$ also follows similarly.
The result \eqref{eq:Qi} then follows.
\end{proof}
\begin{rem}
The result in Theorem~\ref{thm:info} shows that the information matrix(i.e., the inverse of covariance matrix) from Algorithm~\ref{BP} exponentially converges under Assumption~\ref{asu:1}, and the convergence rate is $\rho<1$.
\end{rem}

\section{Convergence Analysis for the Estimates}\label{sec:ConvergenceX}

In this section, we proceed to study the convergence of the estimates
$\hat{x}_{i}(k)$. Under Assumption~\ref{asu:1}, we establish a necessary and sufficient condition
for the asymptotic convergence of the estimates. This result is general and non-conservative but requires checking
the stability of a high-dimensional matrix. We then provide a sufficient
condition for convergence of the estimates which can be easily verified
in a distributed fashion, with low computational complexity. As a
by-product, we also provide an alternative proof for the known result
that the estimates always converge for graphs with at most a single
cycle~\cite{Weiss2000graph}.

From (\ref{eq:BP-Q}) and (\ref{eq:BP-Q1}), we get
\[
\alpha_{i\rightarrow i,j}(k+1)=C_{i}^{T}R_{i}^{-1}z_{i}+\sum_{w\in\mathcal{N}_{i}\backslash j}\alpha_{i,w\rightarrow i}(k).
\]
Similarly, from (\ref{eq:BP-Q2}) and (\ref{eq:BP-R}), we get
\begin{align*}
\alpha_{i,w\rightarrow i}(k)= & C_{i,w}^{T}R_{i,w\rightarrow i}^{-1}(k)z_{w,i}\\
&-C_{i,w}^{T}R_{i,w\rightarrow i}^{-1}(k)C_{w,i}Q_{w\rightarrow i,w}^{-1}(k)\alpha_{w\rightarrow i,w}(k).
\end{align*}
Combining the above two equations gives the following dynamics:
\begin{align*}
 & \alpha_{i\rightarrow i,j}(k+1)\\
= & \beta_{i\rightarrow i,j}(k)\\
&-\sum_{w\in\mathcal{N}_{i}\backslash j}C_{i,w}^{T}R_{i,w\rightarrow i}^{-1}(k)C_{w,i}Q_{w\rightarrow i,w}^{-1}(k)\alpha_{w\rightarrow i,w}(k),
\end{align*}
where
\begin{eqnarray}\label{betaXXX}
\beta_{i\rightarrow i,j}(k)=C_{i}^{T}R_{i}^{-1}z_{i}+\sum_{w\in\mathcal{N}_{i}\backslash j}C_{i,w}^{T}R_{i,w\rightarrow i}^{-1}(k)z_{w,i}.
\end{eqnarray}
It is easy to check with Algorithm 1 that the above holds for all
$k\ge1$, provided that, in the equation above, we initialize all
$\alpha_{i\rightarrow i,j}(0)=0$.

Defining
\begin{align*}
\tilde{x}_{i\rightarrow i,j}(k) & =Q_{i\rightarrow i,j}^{-1/2}(k)\alpha_{i\rightarrow i,j}(k),\\
b_{i\rightarrow i,j}(k) & =Q_{i\rightarrow i,j}^{-1/2}(k+1)\beta_{i\rightarrow i,j}(k),\\
a_{i\rightarrow i,j}(k) & =C_{j,i}^{T}R_{i,j\rightarrow j}^{-1}(k)C_{i,j}Q_{i\rightarrow i,j}^{-1/2}(k),
\end{align*}
we get the following alternative dynamics:
\begin{align*}
 \tilde{x}_{i\rightarrow i,j}(k+1)=& b_{i\rightarrow i,j}(k)-Q_{i\rightarrow i,j}^{-1/2}(k+1)\\
 &\cdot\sum_{w\in\mathcal{N}_{i}\backslash j}a_{w\rightarrow i,w}(k)\tilde{x}_{w\rightarrow i,w}(k).
\end{align*}
Let $S$ be any ordered sequence of all $(i\rightarrow i,j)$. Form
the column vector $\tilde{x}(k)$ by stacking up all the $\tilde{x}_{i\rightarrow i,j}(k)$
according to $S$, and similarly form $b(k)$ by stacking up all the
$b_{i\rightarrow i,j}(k)$.

For each $(i\rightarrow i,j)$, define the row vector $A_{i\rightarrow i,j}(k)$
with its $(w\rightarrow i,w)$-th element equal to $-Q_{i\rightarrow i,j}^{-1/2}(k+1)a_{w\rightarrow i,w}(k)$
for each $w\in\mathcal{N}_{i}\backslash j$, and all other elements
zero. That is,
\begin{align*}
 & A_{i\rightarrow i,j}(k)\tilde{x}(k)\\
 =&-Q_{i\rightarrow i,j}^{-1/2}(k+1)\sum_{w\in\mathcal{N}_{i}\backslash j}a_{w\rightarrow i,w}(k)\tilde{x}_{w\rightarrow i,w}(k).
\end{align*}
Then we have the following dynamics for $\tilde{x}(k)$:
\begin{equation}
\tilde{x}(k+1)=A(k)\tilde{x}(k)+b(k),\label{eq:dynamics}
\end{equation}
where $A(k)$ is a matrix formed by stacking up all the row vectors
$A_{i\rightarrow i,j}(k)$ according to $S$. This leads to the following
main result on the convergence of $\tilde{x}(k)$, which in turn guarantees
the convergence of $\hat{x}_{i}(k)$ due to the convergence of $Q_{i\rightarrow i,j}(k)$.
\begin{lem}\label{lem:estimate}
Under Assumption~\ref{asu:1}, the estimate $\tilde{x}(k)$ converges asymptotically
to $(I-A(\infty))^{-1}b(\infty)$ if the matrix $A(\infty)$ is stable
(i.e., all of its eigenvalues are strictly within the unit circle).
Conversely, if $A(\infty)$ is not stable, then for almost all measurements
of $z_{i}$ and $z_{i,j}$, $\tilde{x}(k)$ will diverge as $k\rightarrow\infty$.
\end{lem}
\begin{proof}
It is clear from its definition that $b(k)$ is bounded. Also note
that $A(k)\rightarrow A(\infty)$ as $k\rightarrow\infty$. Thus,
the convergence property for $\tilde{x}(k)$ follows naturally from
the stability of $A(\infty)$.
\end{proof}

We have the following key property for $A(\infty)$.
\begin{lem}\label{lem:prop}
Under Assumption~\ref{asu:1}, the diagonal elements of $A(\infty)$ are zero.
Moreover, for every $(i\rightarrow i,j)$, we have
\[
A_{i\rightarrow i,j}(\infty)A_{i\rightarrow i,j}^{T}(\infty)\le\rho I.
\]
\end{lem}
\begin{proof}
See Appendix~\ref{app:lem:prop}.
\end{proof}

For a given canonical graph $\mathcal{G}$, denote by $\bar{\mathcal{G}}$
the reduced graph obtained by repeatedly removing the leaf nodes until
there are no more leaf nodes (i.e., all the variable nodes are on
a cycle), or $\bar{\mathcal{G}}$ is a singleton (i.e., it contains
a single variable node). Without loss of generality, let the remaining
nodes be $1,2,\ldots,\bar{I}$. Further denote by $\bar{A}(\infty)$
the matrix obtained by removing the rows and columns of $A(\infty)$
associated with indices $(i\rightarrow i,j)$ for $i>\bar{I}$. For
the case $\bar{\mathcal{G}}$ is a singleton, $\bar{A}(\infty)$ is
void. We have the following important result:
\begin{lem}\label{lem:prop2}
The matrix $A(\infty)$ is stable if and only if
$\bar{A}(\infty)$ is stable. In particular, $A(\infty)$ is always
stable if $\mathcal{G}$ is acyclic.
\end{lem}
\begin{proof}
See Appendix~\ref{app:lem:prop2}.
\end{proof}

Using the result above, we can restate Lemma~\ref{lem:estimate} as
follows:
\begin{thm}
\label{thm:estimate.1} Under Assumption~\ref{asu:1}, every estimate $\hat{x}_i(k)$, $i=1,2,\ldots,I$, converges asymptotically if the matrix $\bar{A}(\infty)$ is stable.
Conversely, if $\bar{A}(\infty)$ is not stable, then for almost all
measurements of $z_{i}$ and $z_{i,j}$, $\hat{x}_i(k)$ will diverge
as $k\rightarrow\infty$.
\end{thm}

While the result above provides a necessary and sufficient condition
for the convergence of the estimates, checking $\bar{A}(\infty)$
is not an easy task for a large system. Our next aim is to provide
a sufficient condition for guaranteeing the stability of $\bar{A}(\infty)$
that is easily verifiable in a distributed fashion.

With some abuse of notation, we still denote by $S$ a sequence of
$(i\rightarrow i,j)$ for $\bar{\mathcal{G}}$. In particular, we
will choose $S=\{S_{1},S_{2},\ldots,S_{\bar{I}}\}$ with $S_{i}$
denoting a sub-sequence containing all $(i\rightarrow i,j)$ for $j\in\bar{\mathcal{N}}_{i}$
and $\bar{\mathcal{N}}_{i}$ denoting the set of neighbouring nodes of $i$
in $\bar{\mathcal{G}}$. We further denote by $\bar{A}_{i}(\infty)$
the square sub-matrix of $\bar{A}(\infty)$ by keeping only the rows
with indices $(i\rightarrow i,j),j\in\bar{\mathcal{N}}_{i}$, and
only columns with indices $(j\rightarrow j,i),j\in\bar{\mathcal{N}}_{i}$.
For better understanding of the notation, an example is given below.

\begin{exa}\label{exam} Fig.~\ref{example1} shows a simple canonical
graph and the structure of the associated $A(\infty)$, where $*$
stands for a non-zero term. It is easy to verify that
\begin{align*}
A_{1}(\infty) & =A_{3}(\infty)=\begin{bmatrix}0 & * & *\\
* & 0 & *\\
* & * & 0
\end{bmatrix};\\
A_{2}(\infty) & =A_{4}(\infty)=\begin{bmatrix}0 & *\\
* & 0
\end{bmatrix}.%\\
%A(\infty)A^{T}(\infty) & =\text{diag}\{A_{1}(\infty)A_{1}^{T}(\infty),\ldots,A_{4}(\infty)A_{4}^{T}(\infty)\}.
\end{align*}
%$\begin{bmatrix}0&0&0&0&0&*&0&0&*&0\\ 0&0&0&*&0&0&0&0&*&0\\ 0&0&0&*&0&*&0&0&0&0\\ 0&0&0&0&0&0&*&0&0&0\\ *&0&0&0&0&0&0&0&0&0\\ 0&0&0&0&*&0&0&0&0&*\\ 0&*&0&0&0&0&0&0&0&*\\ 0&*&0&0&*&0&0&0&0&0\\ 0&0&0&0&0&0&0&*&0&0\\ 0&0&*&0&0&0&0&0&0&0\end{bmatrix}$

\end{exa}
\begin{figure}[ht]
\begin{centering}
\includegraphics[width=9cm]{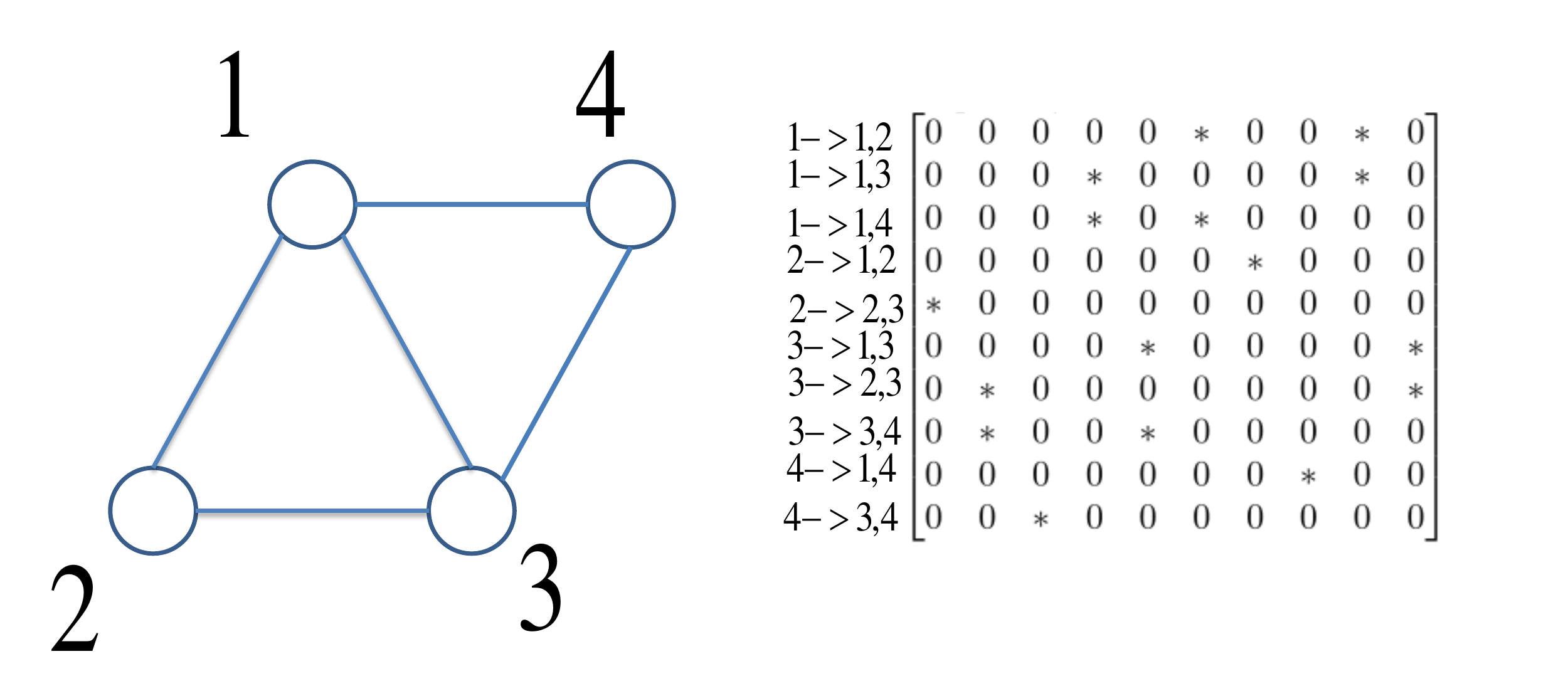}
\par\end{centering}

\caption{An example to show the structure of $A(\infty)$}
\label{example1}
\end{figure}

We have the next result on the convergence of the estimates, which
can be checked in a distributed fashion.
\begin{thm}
\label{thm:estimate2} Recall that $\rho$ is defined in\eqref{rho}. Suppose there exists $\rho\leq \bar{\rho}<1$
such that $\bar{A}_{i}(\infty)\bar{A}_{i}^{T}(\infty)\le\bar{\rho}I$
for every $1\le i\le\bar{I}$ which has at least three neighbouring nodes
in $\bar{\mathcal{G}}$. Under Assumption~\ref{asu:1}, for every $i=1,2,\ldots,I$, $\hat{x}_i(k)$
converges exponentially with the rate $\bar{\rho}$.
In particular, if $\mathcal{G}$ only has a single cycle (which means that
every node $i$ in $\bar{\mathcal{G}}$ has only two neighbouring
nodes), $\hat{x}_i(k)$ converges exponentially with the rate of $\rho$.
\end{thm}

\begin{proof}
Suppose such $\bar{\rho}$ exists. From Lemma~\ref{lem:estimate}
and Lemma~\ref{lem:prop2}, it suffices to show that $\bar{A}^{T}(\infty)\bar{A}(\infty)\le\bar{\rho}I$.
Using Schur complement, this is equivalent to $\bar{A}(\infty)\bar{A}^{T}(\infty)\le\bar{\rho}I$.
It remains to show that this holds if $\bar{A}_{i}(\infty)\bar{A}_{i}^{T}(\infty)\le\bar{\rho}I$
for every $1\le i\le\bar{I}$ which has at least three neighbouring nodes
in $\bar{\mathcal{G}}$.

Indeed, from the construction of $A_{i\rightarrow i,j}(k)$ and definition
of $\bar{A}(\infty)$, it can be deduced that each row $\bar{A}_{i\rightarrow i,j}(\infty)$
of $\bar{A}(\infty)$ has non-zero entries only in column locations
$(w\rightarrow w,i)$ with $w\in\bar{\mathcal{N}}_{i}\backslash j$.
Two results follows from this. Firstly, for any $i\ne u$, we have
$\bar{A}_{i\rightarrow i,j}(\infty)\bar{A}_{u\rightarrow u,v}^{T}(\infty)=0$
because $\bar{A}_{i\rightarrow i,j}(\infty)$ and $\bar{A}_{u\rightarrow u,v}(\infty)$
have no common non-zero entries. Secondly, if any node $i$ in $\bar{\mathcal{G}}$
has only two neighboring nodes, say $j$ and $t$, then $\bar{A}_{i\rightarrow i,j}(\infty)$
has only one nonzero entry in column $(t\rightarrow t,i)$ and, likewise,
$\bar{A}_{i\rightarrow i,t}(\infty)$ has only one nonzero entry in
column $(j\rightarrow j,i)$. This means that $\bar{A}_{i\rightarrow i,j}(\infty)\bar{A}_{i\rightarrow i,t}^{T}(\infty)=0$
as well. Please refer to Example~\ref{exam} for the better understanding
of two facts above. Using the first result above, we have
\begin{align*}
&\bar{A}(\infty)\bar{A}^{T}(\infty)\\
=&\mathrm{diag}\{\bar{A}_{1}(\infty)\bar{A}_{1}^{T}(\infty),\bar{A}_{2}(\infty)\bar{A}_{2}^{T}(\infty),\ldots\}.
\end{align*}
The second result above further implies that if node $i$ has only
two neighbouring nodes in $\bar{\mathcal{G}}$, the corresponding $\bar{A}_{i}(\infty)\bar{A}_{i}^{T}(\infty)$
is a diagonal matrix (actually it is a $2\times2$ matrix), and by
Lemma~\ref{lem:prop}, $\bar{A}_{i}(\infty)\bar{A}_{i}^{T}(\infty)\le\rho I$.
Since $\rho\le\bar{\rho}$ and that $\bar{A}_{i}(\infty)\bar{A}_{i}^{T}(\infty)\le\bar{\rho}I$
for every $1\le i\le\bar{I}$ which has at least three neighbouring nodes
in $\bar{\mathcal{G}}$, we conclude that $\bar{A}(\infty)\bar{A}^{T}(\infty)\le\bar{\rho}I$.

The convergence rate of $\rho$ for the case of $\mathcal{G}$ having
a single cycle is clear from the above discussion as well because every
$\bar{A}_{i}(\infty)\bar{A}_{i}^{T}(\infty)\le\rho I$ in this case. \end{proof}

\begin{rem}
Using Theorem~\ref{thm:estimate2}, checking the convergence of the
estimates amounts to computing $\bar{A}_{i}(\infty)$ and its maximum
singular value $\sigma_{i}$ for each node with at least three neighbouring
nodes. The required $\bar{\rho}$ can be made to be $\bar{\rho}=\max_{i}\sigma_{i}^{2}$
using the fact that $\bar{A}_{i}(\infty)\bar{A}_{i}^{T}(\infty)\le\sigma_{i}^{2}I$.
By Theorem~\ref{thm:estimate2}, the convergence of the estimates
is guaranteed if $\bar{\rho}<1$. %Note that $\bar{A}_i(\infty)$ can be computed by each node $i$ in a distributed fashion using only neighbourhood information (i.e., information from itself and its neighbouring nodes). The value of $\bar{\rho}$ can be obtained by applying a max consensus algorithm ({\bf \textcolor{red}{ Need a reference})} on $\sigma_i^2$, which is a distributed algorithm. Hence,  Theorem~\ref{thm:estimate2} provides a distributed method for checking the  convergence of the estimates.
\end{rem}
\begin{rem}\label{rem:dynamic}
The Algorithm~\ref{BP} is designed for the static estimation problem. As for the centralized(traditional) state estimation case, this is a crucial step towards dynamic state estimation. Generalization to distributed dynamic state estimation will be a future topic. Since we have proved that the estimate using our algorithm exponentially converges, the algorithm could achieve an approximation for the optimal state estimate at time $k$ in a few steps during the time update from $k$ to $k+1$ and the prediction for time $k+1$ follows from neighbours' state estimates.
\end{rem}

\section{Accuracy Analysis on Information Matrices}\label{sec:AccuracyQ}

In Section~\ref{sec:ConvergenceQ}, we studied the convergence of
the information matrices. In this section we study its accuracy, i.e.,
the difference between the information matrices $Q_{i}(k)$ generated by Algorithm~\ref{BP} and the information matrices for the maximum likelihood estimates.

Let $\mathcal{G}_{i}(d)$ denote the subgraph of $\mathcal{G}$ formed
by nodes which are within $d$ hops away from node~$i$. Denote by
$d_{i}$ the largest integer such that $\mathcal{G}_{i}(d_{i})$ is
acyclic. We refer to $d_{i}$ as the {\em cycle-free depth} of node
$i$, and $d_{\min}=\min_{i}d_{i}$ as the {\em cycle-free depth}
of $\mathcal{G}$. If $\mathcal{G}$ is acyclic, we use the convention
that $d_{i}=\infty$ for all $i$, and in this case, $\mathcal{G}_{i}(d_{i})=\mathcal{G}$. It is emphasized that the {\em cycle-free depth} is a property of each vertex, and that it is upper bounded by the diameter of the graph and lower bounded by the length of the shortest cycle in the graph.

Recall that, for the measurements (\ref{eq:cycsys1}) and (\ref{eq:cycsys2}),
the corresponding joint likelihood function $f(X)$ is given by (\ref{Gaussian_fi}), \eqref{Gaussian_fij} and (\ref{Gaussian_f}).
The marginal $g_{i}(x_{i})$ of $X$ is given by (\ref{marginal}).
Denote by $\hat{x}_{i}^{ML}$ and $\Sigma_{i}^{ML}$ the mean and
covariance of $x_{i}$ corresponding to $g_{i}(x_{i})$, respectively.
The superscript ``ML'' stands for maximum likelihood, due to the
fact that $g_{i}(x_{i})$ is Gaussian and thus $\hat{x}_{i}^{ML}$
is the maximum likelihood estimate of $x_{i}$. Also define the ML
information matrix $Q_{i}^{ML}=(\Sigma_{i}^{ML})^{-1}$.

\begin{figure}[ht]
\begin{centering}
\includegraphics[width=8cm]{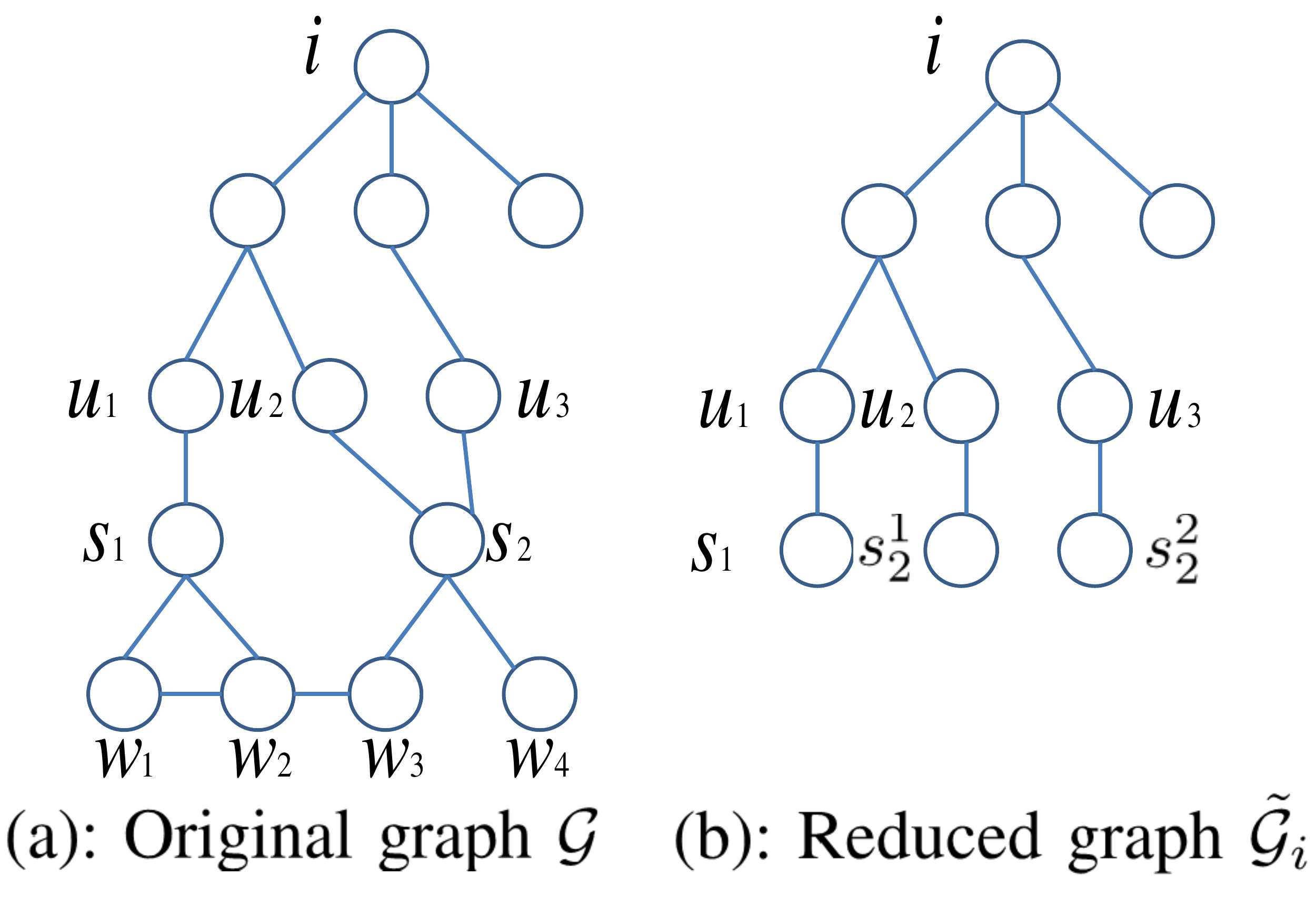}
\par\end{centering}

\caption{Conversion of a cyclic graph into an acyclic graph}
\label{tree}
\end{figure}

%(a): Original graph $\mathcal{G}$%(b): Reduced graph $\tilde{\mathcal{G}}_i$

Now, for any node $i\in\mathcal{G}$, we introduce a reduced graph
$\tilde{\mathcal{G}}_{i}$, which will hold the key to the accuracy
analysis. Draw the graph $\mathcal{G}$ as depicted in Fig.~\ref{tree}(a).
Let $d_{i}$ be the cycle-free depth for node $i$, as defined earlier.
In Fig.~\ref{tree}(a), $d_{i}=2$. Denote by $u_{1},u_{2},\ldots,$
the leaf nodes of $\mathcal{G}_{i}(d_{i})$ and denote by $s_{1},s_{2},\ldots,$
the nodes outside of $\mathcal{G}_{i}(d_{i})$ which are connected
to the leaf nodes and called {\em child} nodes. There
are two cases for each child node: It connects to either one leaf
node only or multiple leaf nodes. In Fig.~\ref{tree}(a), $s_{1}$
connects to one leaf node of $\mathcal{G}_{i}(d_{i})$(i.e., $u_{1}$)
and $s_{2}$ connects to two leaf nodes of $\mathcal{G}_{i}(d_{i})$(i.e.,
$u_{2}$ and $u_{3}$). The reduced graph $\tilde{\mathcal{G}}_{i}$
is depicted in Fig.~\ref{tree}(b). This is done as follows: For
each child node $s$, firstly remove all of its connecting nodes that
are not in $\mathcal{G}_{i}(d_{i}+1)$. Then, if $s$ is connected
to multiple leaf nodes of $\mathcal{G}_{i}(d_{i})$, split $s$ into
multiple copies, one for each connecting leaf node. In Fig.~\ref{tree}(b),
$w_{1},\ldots,w_{4}$ are all removed and $s_{2}$ is split into $s_{2}^{1}$
and $s_{2}^{2}$. For each child node $s_{i}$ that connects to only
one leaf node of $\mathcal{G}_{i}(d_{i})$, let its self measurement
in the reduced graph $\tilde{\mathcal{G}}_{i}$ be $z_{s_{i}}=C_{s_{i}}x_{s_{i}}+v_{s_{i}}$
with noise covariance $R_{s_{i}}$ for $v_{s_{i}}$. Then, for each
child node $s_{i}^{t}$ that connects to $p$ leaf nodes of $\mathcal{G}_{i}(d_{i})$,
we take
\[
z_{s_{i}^{t}}=C_{s_{i}}x_{s_{i}^{t}}+v_{s_{i}^{t}}=z_{s_{i}}
\]
with noise covariance $pR_{s_{i}}$ for $v_{s_{i}^{t}}$, where $p$
is the number of connecting leaf nodes of $\mathcal{G}_{i}(d_{i})$
for $s_{i}$.

We have the following result, which shows that the information
matrices from Algorithm~\ref{BP} converge to ML information matrices exponentially as the
cycle-free depths increase.

\begin{thm}
\label{thm:accuracyQ} Under Assumption~\ref{asu:1}, for
every node $i$ in $\mathcal{G}$, we have
\begin{equation}
0\le Q_{i}(d_i)-Q_{i}^{ML}\le\tilde{\alpha}_{i}\tilde{\rho}_{i}^{d_{i}-1}Q_{i}^{ML},\label{eq:accuracyQ}
\end{equation}
where $\tilde{\alpha}_{i}$ and $\tilde{\rho}_{i}$ are similar to
$\alpha$ in \eqref{alpha0} and $\rho$ in \eqref{rho}, but for the reduced graph
$\tilde{\mathcal{G}}_{i}$.
\end{thm}
\begin{proof}
Taking the graph in Fig.~\ref{tree} as an example. Firstly, we introduce two trimmed graphs $\hat{\mathcal{G}}^{1}$
and $\hat{\mathcal{G}}^{2}$ generated from graph $\mathcal{G}$,
as depicted in Fig.~\ref{tree3}(a) and Fig.~\ref{tree3}(b).
\begin{figure}[ht]
\begin{centering}
\includegraphics[width=8cm]{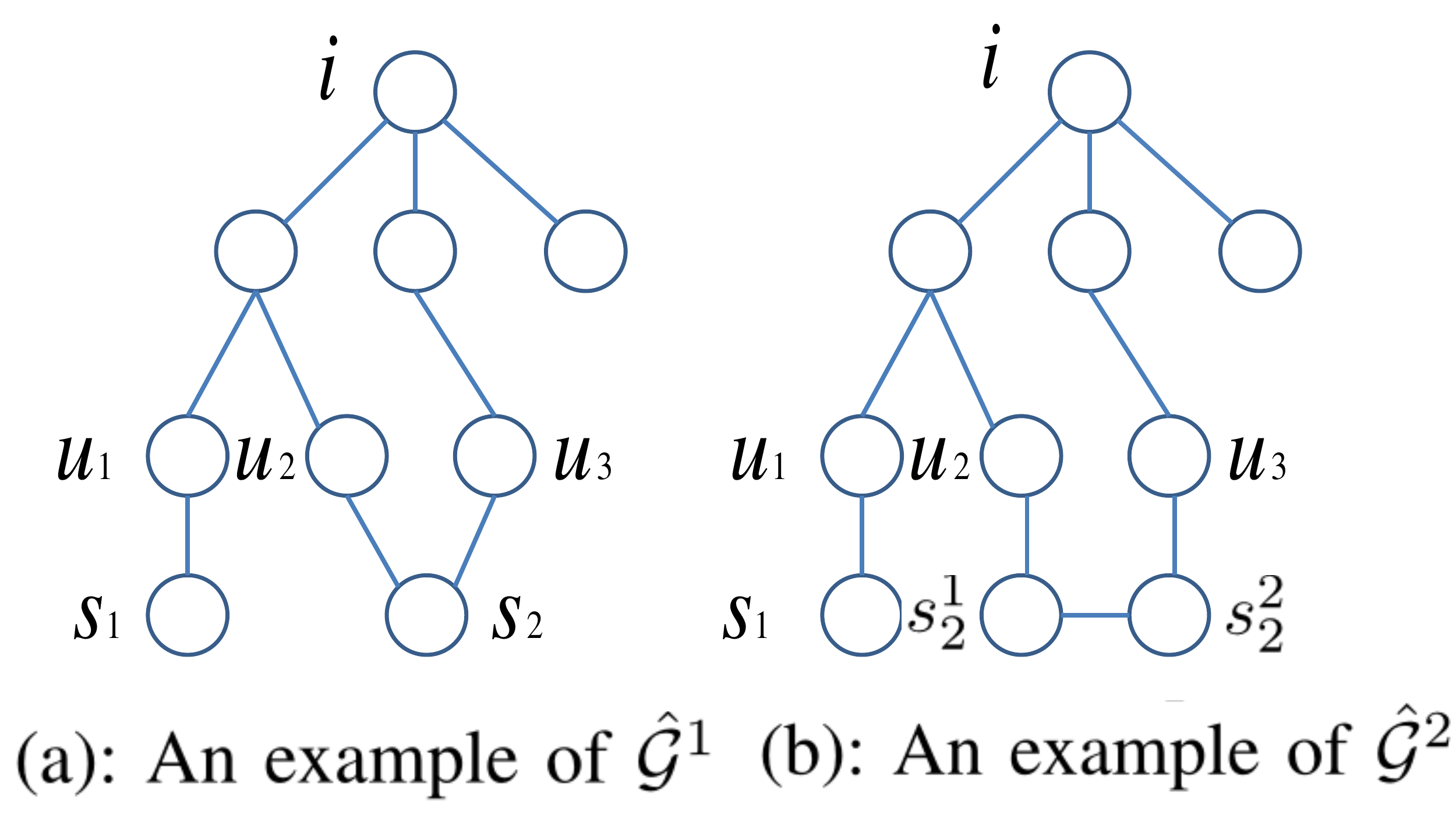}
\par\end{centering}

\caption{Example graphs $\hat{\mathcal{G}}^{1}$ and $\hat{\mathcal{G}}^{2}$}
\label{tree3}
\end{figure}

The Fig.~\ref{tree3}(a) is generated from Fig.~\ref{tree}(a) by
cutting all the nodes and connections outside $\mathcal{G}(d_{i}+1)$.
While, Fig.~\ref{tree3}(b)'s only difference with Fig.~\ref{tree}(b)
is connecting each of $s_{i}^{t}$. Additionally, the edge(joint) measurement
for $(s_{i}^{t},s_{i}^{t+1})$ is taken to be
\[
z_{s_{i}^{t},s_{i}^{t+1}}=x_{s_{i}^{t}}-x_{s_{i}^{t+1}}+v_{s_{i}^{t},s_{i}^{t+1}}
\]
with noise covariance $R_{s_{i}^{t},s_{i}^{t+1}}$ to be discussed
later.

We claim that if $R_{s_{i}^{t},s_{i}^{t+1}}=0$ for all $t$ and $i$,
then the graph $\hat{\mathcal{G}}^{1}$(i.e., Fig.~\ref{tree3}(a))
is identical to the graph $\hat{\mathcal{G}}^{2}$(i.e., Fig.~\ref{tree3}(b)).
Indeed, $R_{s_{i}^{t},s_{i}^{t+1}}=0$ for all $t$ will force all
$x_{s_{i}^{t}}$, $t=1,2,\ldots,p$, to be identical. That is, all
the nodes $s_{i}^{t}$ with the same $i$ can be merged into one node.
From the construction of the self and edge(joint) measurements associated
with $s_{i}^{t}$, we see that this merged node has the same combined
noise covariances as the original node $s_{i}$ has. Hence, our claim
holds.

Furthermore, by letting $R_{s_{i}^{t},s_{i}^{t+1}}=\infty$ for all
$t$ and $i$, the graph $\hat{\mathcal{G}}^{2}$(i.e., Fig.~\ref{tree3}(b))
is identical to the graph $\tilde{\mathcal{G}}_i$(i.e., Fig.~\ref{tree}(b)).
Denote by $\hat{Q}_{i}^{ML}$ and $\tilde{Q}_i^{ML}$ the information matrix for node $i$
by running maximum likelihood estimation on $\hat{\mathcal{G}}^{2}$ and $\tilde{\mathcal{G}}_i$(i.e.,
Fig.~\ref{tree3}(b) and Fig.~\ref{tree}(b)). We note that, as all $R_{s_{i}^{t},s_{i}^{t+1}}$
increase from 0, the corresponding edge measurements become less accurate,
which will cause $\hat{Q}_{i}^{ML}$ to decrease. That is, $\hat{Q}_{i}^{ML}\geq\tilde{Q}_{i}^{ML}$.
Moreover, comparing with the graph $\mathcal{G}$(i.e., Fig.~\ref{tree}(a)),
graph $\hat{\mathcal{G}}^{1}$(i.e., Fig.~\ref{tree3}(a)) has less
measurements and it follows that ${Q}_{i}^{ML}\geq\hat{Q}_{i}^{ML}$.

Denote by $\tilde{Q}_{i}(k)$ the information matrix of node~$i$ at time $k$ by running the Algorithm~\ref{BP} on the graph $\tilde{\mathcal{G}}_i$. Comparing the graph $\mathcal{G}$(i.e., Fig.~\ref{tree}(a))
and $\tilde{\mathcal{G}}_i$(i.e., Fig.~\ref{tree}(b)), it is straightforward
to get that $\tilde{Q}_{i}(k)\leq{Q}_{i}(k)$ for all $k$.

From the discussions above, we have $\tilde{Q}_{i}(\infty)\leq Q_{i}(\infty)$
and $\tilde{Q}_{i}^{ML}\leq Q_{i}^{ML}$.

Because $\tilde{\mathcal{G}}_{i}$ is acyclic, we have $\tilde{Q}_{i}^{ML}=\tilde{Q}_{i}(\infty)$.
It is also clear that $Q_{i}(d_{i})=\tilde{Q}_{i}(d_{i})$ because
both $\mathcal{G}$ and $\tilde{\mathcal{G}}_{i}$ have the same structure
within $d_{i}$ hops away from node $i$. Furthermore, from Theorem~\ref{thm:info},
we have
\[
0\le\tilde{Q}_{i}(d_{i})-\tilde{Q}_{i}(\infty)\le\tilde{\alpha}_{i}\tilde{\rho}_{i}^{d_{i}-1}\tilde{Q}_{i}(\infty).
\]
Combining the above, it follows that
\begin{align}
&Q_{i}(d_{i})-Q_{i}^{ML} \nonumber\\
=&(\tilde{Q}_{i}(d_{i})-\tilde{Q}_{i}(\infty))+(\tilde{Q}_{i}(\infty)-Q_{i}^{ML})\nonumber \\
\le&\tilde{\alpha}_{i}\tilde{\rho}_{i}^{d_{i}-1}\tilde{Q}_{i}(\infty)+(\tilde{Q}_{i}^{ML}-Q_{i}^{ML})\nonumber \\
\le&\tilde{\alpha}_{i}\tilde{\rho}_{i}^{d_{i}-1}\tilde{Q}_{i}^{ML}\nonumber \\
\le&\tilde{\alpha}_{i}\tilde{\rho}_{i}^{d_{i}-1}Q_{i}^{ML}.\label{tempx}
\end{align}

We now return to the original graph $\mathcal{G}$ and claim that
$Q_{i}(d_{i})$ can be interpreted as the maximum likelihood information
matrix for node $i$ when $R_{s}\rightarrow0$ for all child node
$s$. Indeed, this is the same as making $x_{s}$ perfectly known.
It follows that each leaf node $u$ of $\mathcal{G}_{i}(d_{i})$ connected
to $s$ will have prior information matrix for $x_{u}$ equal to $C_{u,s}^{T}R_{u,s}^{-1}C_{u,s}$,
which is exactly the initialization step for $Q_{u,s\rightarrow u}(0)$
in (\ref{newinitial}). Hence our claim holds. Now, since the maximum
likelihood information matrix increases as $R_{s}$ decreases, we
conclude that $Q_{i}(d_{i})\ge Q_{i}^{ML}$. Combining this with (\ref{tempx}),
we get \eqref{eq:accuracyQ}.
\end{proof}

Using Theorem~\ref{thm:info}, the result in Theorem~\ref{thm:accuracyQ} can be stated in a different way in terms of the accuracy of $Q_i(\infty)$.

\begin{cor}\label{cor:Q_inf_acc}
Recall that $\tilde{\alpha}_{i},\tilde{\rho}_{i},\alpha,\rho$ are defined in the Theorem~\ref{thm:accuracyQ} and $d_{i}$ is the cycle-free depth of node~$i$. Under Assumption~\ref{asu:1}, we have, for every node $i$ in $\mathcal{G}$,
\begin{equation}
-\frac{\alpha \rho^{d_i-1}}{1+\alpha \rho^{d_i-1}} Q_i^{ML} \le Q_i(\infty)-Q_i^{ML}
\le \tilde{\alpha}_{i}\tilde{\rho}_{i}^{d_{i}-1}Q_{i}^{ML}. \label{Q_inf_acc}
\end{equation}
In particular, as $d_i\rightarrow\infty$, $Q_i(\infty)\rightarrow Q_i^{ML}$.
\end{cor}
\begin{proof}
Using the monotonicity property of $Q_i(k)$ (Lemma~\ref{lem:01}) and Theorem~\ref{thm:accuracyQ}, we have
\[
Q_i(\infty) - Q_i^{ML} \le Q_i(d_i)-Q_i^{ML} \le \tilde{\alpha}_{i}\tilde{\rho}_{i}^{d_{i}-1}Q_{i}^{ML}
\]
which is the right hand side of (\ref{Q_inf_acc}). On the other hand, using Theorem~\ref{thm:info} and Theorem~\ref{thm:accuracyQ}, we get
\begin{align*}
&Q_i(\infty) - Q_i^{ML} \\
=& (Q_i(\infty) - Q_i(d_i))+(Q_i(d_i)-Q_i^{ML} )\\
\ge&-\alpha\rho^{d_i-1}Q_i(\infty)
\end{align*}
which gives
\[ Q_i(\infty) \ge (1+\alpha \rho^{d_i-1})^{-1}Q_i^{ML}.\]
Then, it deduces the left hand side of (\ref{Q_inf_acc}).
\end{proof}

\begin{rem}
\label{rem:accuracyQ} We mention two properties about the $\tilde{\alpha}_{i}$ and $\tilde{\rho}_{i}$. Firstly, it is clear that for an acyclic
graph $\mathcal{G}$, $\tilde{\alpha}_{i}=\alpha$ and $\tilde{\rho}_{i}=\rho$
for any $i$. Secondly, since each reduced
graph is for a given node $i$ and is typically a small graph (with
cycle-free depth of $d_{i}+1$), $\tilde{\alpha}_{i}$ and $\tilde{\rho}_{i}$
can be computed quickly (with $d_{i}+1$ iterations).
\end{rem}

\section{Accuracy Analysis for the Estimates}\label{sec:AccuracyX}
In this section, we study the accuracy of the
state estimate from Algorithm~\ref{BP}. Our goal is to characterize explicitly how
the distributed state estimate accuracy for node $i$ is related to its cycle-free
depth $d_{i}$.

Without loss of generality and for notational simplicity, we study
node 1 in this section. Let $d_{1}$ be the cycle-free depth of node~1. We can redraw
the original graph $\mathcal{G}$ (i.e., Fig.~\ref{tree}(a)) as a $d_{1}+2$ layer graph in Fig.~\ref{tree2}(a),
in which node 1 is placed on the top layer (layer 1), followed by
all the nodes one lop away from node 1 as layer 2, then by all the
nodes two hops away from node 1 as layer 3, and so on, until layer
$(d_{1}+1)$ which contains all the nodes $d_{1}$ hops away from
node 1. All other nodes (i.e., nodes outside of $\mathcal{G}_{1}(d_{1})$)
are lumped into layer $d_{1}+2$. This graph can then be redrawn again
as a line graph $\mathcal{G}^{l}$(i.e., Fig.~\ref{tree2}(b)) by
grouping all the nodes in layer $i$ as a super node $i$ (denoted
by SN $i$) with state $\check{x}_{i}$ which is formed by stacking
up all the states in layer $i$. In particular, $\check{x}_{1}=x_{1}$.
The self measurement for super node $i$ will be denoted by $\check{z}_{i}$,
$i=1,2,\ldots,d_{i}+2$, and the edge measurement between super nodes
$i$ and $i+1$ will be denoted by $\check{z}_{i,i+1}$, $i=1,2,\ldots,d_{i}+1$.
The notations of $\check{C}_{i},\check{C}_{i,j},\check{R}_{i}$ and
$\check{R}_{i,j}$ are similarly defined. For notational simplicity,
we denote $d_{1}+2$ by $n$.

\begin{figure}[ht]
\begin{centering}
\includegraphics[width=8cm]{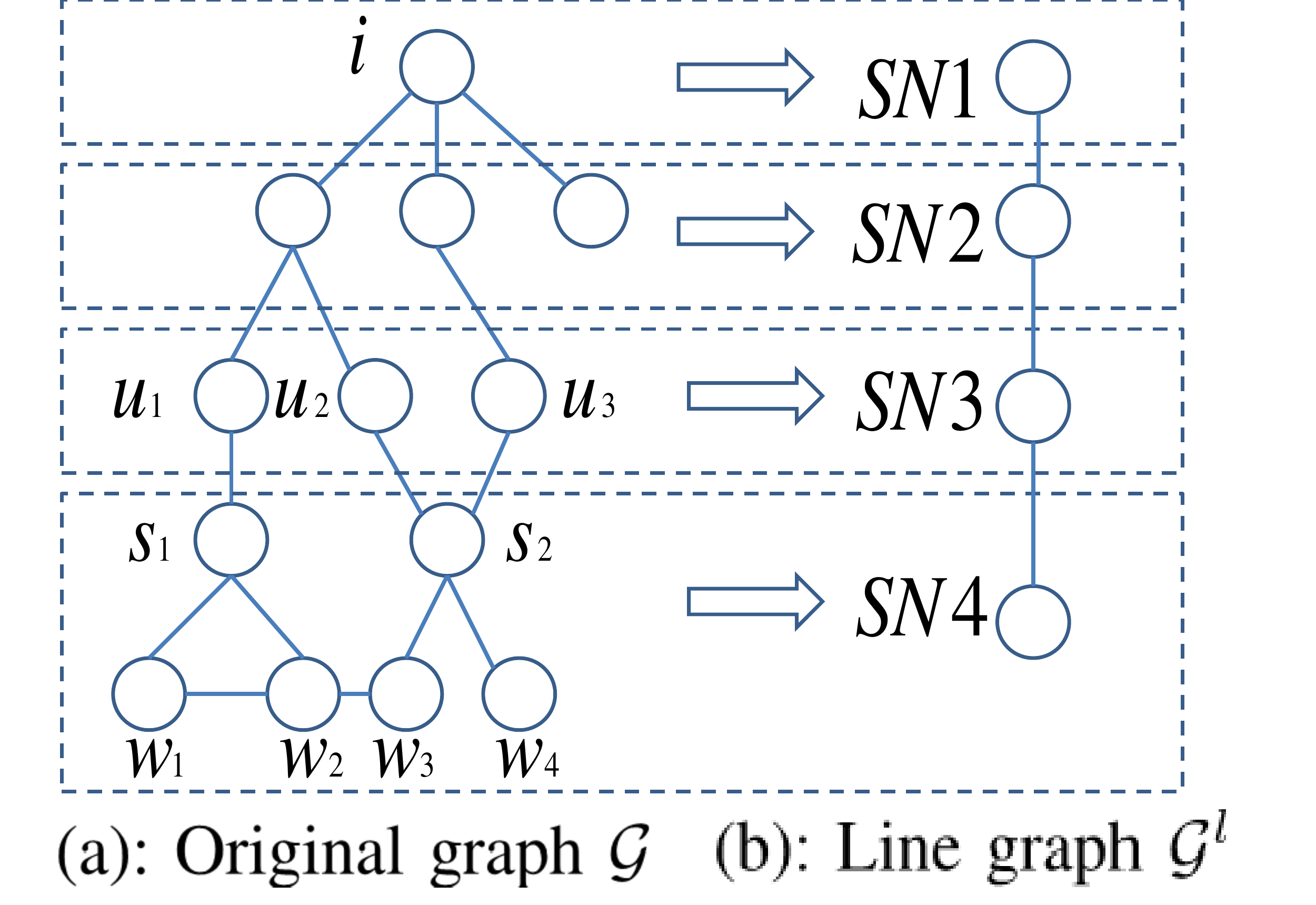}
\par\end{centering}
\caption{Conversion of a graph with $d_{1}=2$ into a 4-layer line graph}
\label{tree2}
\end{figure}

Denoting $\check{x}=\mathrm{col}\{\check{x}_{1},\check{x}_{2},\ldots,\check{x}_{n}\}$,
we have the following result on the maximum likelihood estimate for
$\check{x}$ in the trimmed line graph $\mathcal{G}^{l}$(i.e., Fig.~\ref{tree2}(b)).
%(b): Line graph $\mathcal{G}^l$
\begin{lem}
\label{lem:tree2} The maximum likelihood estimate $\check{x}^{ML}$
for $\check{x}$ is given by the solution of
\begin{equation}
A\check{x}^{ML}=B,\label{eq:AB}
\end{equation}
where $A>0$ is a tri-diagonal block matrix defined by
\begin{align*}
A_{11} & =\check{C}_{1}^{T}\check{R}_{1}^{-1}\check{C}_{1}+\check{C}_{1,2}^{T}\check{R}_{1,2}^{-1}\check{C}_{1,2}\\
A_{ii} & =\check{C}_{i}^{T}\check{R}_{i}^{-1}\check{C}_{i}+\check{C}_{i,i+1}^{T}\check{R}_{i,i+1}^{-1}\check{C}_{i,i+1}\\
&~~+\check{C}_{i,i-1}^{T}\check{R}_{i-1,i}^{-1}\check{C}_{i,i-1},\ \ \ i=2,\ldots,n-1\\
A_{nn} & =\check{C}_{n}^{T}\check{R}_{n}^{-1}\check{C}_{n}+\check{C}_{n,n-1}^{T}\check{R}_{n-1,n}^{-1}\check{C}_{n,n-1}\\
A_{i(i+1)} & =\check{C}_{i,i+1}^{T}\check{R}_{i,i+1}^{-1}\check{C}_{i+1,i},\ \ \ i=1,2,\ldots,n-1\\
A_{(i+1)i} & =A_{i(i+1)}^{T}\\
A_{ij} & =0,\ \ \ |i-j|>1,
\end{align*}
and $B=\mathrm{col}\{B_{1},B_{2},\ldots,B_{n}\}$ with
\begin{align*}
B_{1} & =\check{C}_{1}^{T}\check{R}_{1}^{-1}\check{z}_{1}+\check{C}_{1,2}^{T}\check{R}_{1,2}^{-1}\check{z}_{1,2}\\
B_{i} & =\check{C}_{i}^{T}\check{R}_{i}^{-1}\check{z}_{i}+\check{C}_{i,i-1}^{T}\check{R}_{i-1,i}^{-1}\check{z}_{i-1,i}\\
&~~+\check{C}_{i,i+1}^{T}\check{R}_{i,i+1}^{-1}\check{z}_{i,i+1},\ \ \ i=2,\ldots,n-1\\
B_{n} & =\check{C}_{n}^{T}\check{R}_{n}^{-1}\check{z}_{n}+\check{C}_{n,n-1}^{T}\check{R}_{n-1,n}^{-1}\check{z}_{n,n-1}.
\end{align*}
\end{lem}

\begin{proof}
Define
\[
F(\check{x})=\sum_{i=1}^{n}F_{i}(\check{x}_{i})+\sum_{i=1}^{n-1}F_{i,i+1}(\check{x}_{i},\check{x}_{i+1}),
\]
where
\begin{align*}
F_{i}(\check{x}_{i}) & =(\check{C}_{i}\check{x}_{i}-\check{z}_{i})^{T}\check{R}_{i}^{-1}(\check{C}_{i}\check{x}_{i}-\check{z}_{i})\\
F_{i,i+1}(\check{x}_{i},\check{x}_{i+1}) & =(\check{C}_{i,i+1}\check{x}_{i}+\check{C}_{i+1,i}\check{x}_{i+1}-\check{z}_{i,i+1})^{T}\\
&\cdot \check{R}_{i,i+1}^{-1}(\check{C}_{i,i+1}\check{x}_{i}+\check{C}_{i+1,i}\check{x}_{i+1}-\check{z}_{i,i+1}).
\end{align*}
Then, the maximum likelihood estimate is obtained by minimizing $F(\check{x})$,
i.e., by setting
\begin{align*}
0 & =\frac{\partial}{\partial\check{x}_{i}}F(\check{x})\\
 & =\check{C}_{i}^{T}\check{R}_{i}^{-1}(\check{C}_{i}\check{x}_{i}-\check{z}_{i})\\
 & \ \ +\check{C}_{i,i+1}^{T}\check{R}_{i,i+1}^{-1}(\check{C}_{i,i+1}\check{x}_{i}+\check{C}_{i+1,i}\check{x}_{i+1}-\check{z}_{i,i+1})\\
 & \ \ +\check{C}_{i,i-1}^{T}\check{R}_{i-1,i}^{-1}(\check{C}_{i-1,i}\check{x}_{i-1}+\check{C}_{i,i-1}\check{x}_{i}-\check{z}_{i-1,i})
\end{align*}
for $i=1,2,\ldots,n$ (For $i=1$ the third term is void and for $i=n$
the second term is void). Reorganizing the equations above gives (\ref{eq:AB}).
\end{proof}

Next we give an alternative characterization for the state estimate
$\hat{x}_{1}(d_{1})$.
\begin{lem}
\label{lem:tree2_BP} Under Assumption~\ref{asu:1}, the state estimate $\hat{x}_{1}(d_{1})$ from Algorithm~\ref{BP} is given by the first block of $\check{\mathbf{x}}$ which solves
\begin{equation}
\mathbf{A}\check{\mathbf{x}}=\mathbf{B},\label{eq:tree2_BP}
\end{equation}
where $\mathbf{A}$ is obtained from $A$ by removing its last row
block and last column block, and $\mathbf{B}$ is obtained from $B$
by removing its last row block.
\end{lem}
\begin{proof}
See Appendix~\ref{app:lem:tree2_BP}.
\end{proof}

The next result characterizes the estimation error of the state estimate $\hat{x}_{1}(d_{1})$.
\begin{lem}
\label{lem:tree2_sol} Let $\Delta x_{1}(d_{1})=\hat{x}_{1}(d_{1})-x_{1}^{ML}$
be the estimation error for node 1. Then, under Assumption~\ref{asu:1}, we have
\begin{align}
&\Delta x_{1}(d_{1})\nonumber\\
=&(-\tilde{A}_{11}^{-1}A_{12})\ldots(-\tilde{A}_{(n-1)(n-1)}^{-1}A_{(n-1)n})\check{x}_{n}^{ML},\label{eq:AB_sol}
\end{align}
where
\begin{align*}
\tilde{A}_{11} & =A_{11},\\
\tilde{A}_{ii} & =A_{ii}-A_{(i-1)i}^{T}\tilde{A}_{(i-1)(i-1)}^{-1}A_{(i-1)i}
\end{align*}
for all $i=2,\ldots,n-1$.
%\end{align*}
\end{lem}
\begin{proof}
See Appendix~\ref{app:lem:tree2_sol}.
\end{proof}

\begin{thm}
\label{thm:accuracy_estimate} Under Assumption \ref{asu:1}, for node $1$ in the graph $\mathcal{G}$ with {cycle-free depth
$d_{1}$}, we have
\begin{equation}
\Delta x_{1}(d_{1})^{T}Q_{1}(1)\Delta x_{1}(d_{1})\le\kappa\eta^{d_{1}},\label{eq:acc_est}
\end{equation}
with $\eta<1$ defined in \eqref{eq:as1} and
\[
\kappa=\sum_{(t,j)}(x_{j}^{ML})^{T}C_{j,t}^{T}R_{t,j}^{-1}C_{j,t}x_{j}^{ML}
\]
where $(t,j)$ are such that node $t$ is $d_{1}$ hops away from
node $1$ and node $j$ is connected to node $t$ but $d_{1}+1$ hops
away from node $1$.
\end{thm}
\begin{proof}
From the Lemma~\ref{lem:tree2_sol}, we have
\begin{align}
 & \Delta x_{1}(d_{1})^{T}A_{11}\Delta x_{1}(d_{1})\nonumber \\
= & (\check{x}_{n}^{ML})^{T}(A_{(n-1)n}^{T}\tilde{A}_{(n-1)(n-1)}^{-1})\ldots(A_{12}^{T}\tilde{A}_{11}^{-1})\nonumber\\
&\cdot \tilde{A}_{11}(\tilde{A}_{11}^{-1}A_{12})\ldots(\tilde{A}_{(n-1)(n-1)}^{-1}A_{(n-1)n})\check{x}_{n}^{ML}.\label{eq:temp10}
\end{align}
We are going to claim that $A_{12}^{T}\tilde{A}_{11}^{-1}A_{12}\le\eta\tilde{A}_{22}$,
where $0<\eta<1$ is specified in (\ref{eq:as1}). Indeed, the above
is equivalent to show
\begin{align*}
&A_{12}^{T}\tilde{A}_{11}^{-1}A_{12} \\
\leq&\eta(\check{C}_{2}^{T}\check{R}_{2}^{-1}\check{C}_{2}+\check{C}_{2,3}^{T}\check{R}_{2,3}^{-1}\check{C}_{2,3}+\check{C}_{2,1}^{T}\check{R}_{1,2}^{-1}\check{C}_{2,1}\\
&~~-A_{12}^{T}\tilde{A}_{11}^{-1}A_{12}).
\end{align*}
Note that $\check{C}_{2,1}^{T}\check{R}_{1,2}^{-1}\check{C}_{2,1}\ge A_{12}^{T}\tilde{A}_{11}^{-1}A_{12}$,
which is due to the fact that
\begin{align*}
 & \left[\begin{array}{cc}
A_{11} & A_{12}\\
A_{12}^{T} & \check{C}_{2,1}^{T}\check{R}_{2,1}^{-1}\check{C}_{2,1}
\end{array}\right]\\
= & \left[\begin{array}{cc}
\check{C}_{1}^{T}\check{R}_{1}^{-1}\check{C}_{1}+\check{C}_{1,2}^{T}\check{R}_{1,2}^{-1}\check{C}_{1,2} & \check{C}_{1,2}^{T}\check{R}_{1,2}^{-1}\check{C}_{2,1}\\
\check{C}_{2,1}^{T}\check{R}_{1,2}^{-1}\check{C}_{1,2} & \check{C}_{2,1}^{T}\check{R}_{1,2}^{-1}\check{C}_{2,1}
\end{array}\right]\\
= & \left[\begin{array}{cc}
\check{C}_{1}^{T}\check{R}_{1}^{-1}\check{C}_{1} & 0\\
0 & 0
\end{array}\right]+\left[\begin{array}{cc}
\check{C}_{1,2}^{T}\\
\check{C}_{2,1}^{T}
\end{array}\right]R_{1,2}^{-1}\left[\begin{array}{cc}
\check{C}_{1,2} & \check{C}_{2,1}\end{array}\right]\\
\ge & ~0.
\end{align*}
Thus, we only need to show that
\begin{align}
\eta(\check{C}_{2}^{T}\check{R}_{2}^{-1}\check{C}_{2}+\check{C}_{2,3}^{T}\check{R}_{2,3}^{-1}\check{C}_{2,3})\ge\check{C}_{2,1}^{T}\check{R}_{2,1}^{-1}\check{C}_{2,1}.\label{eq:temp11}
\end{align}
This property is guaranteed by (\ref{eq:as1}) under Assumption~\ref{asu:1}.
More precisely, for each node $t$ in layer 2, (\ref{eq:as1}) implies
\[
\eta(C_{t}^{T}R_{t}^{-1}C_{t}+\sum_{w\in\mathcal{N}_{t}\backslash1}C_{t,w}^{T}R_{t,w}^{-1}C_{t,w})\ge C_{t,1}^{T}R_{t,1}^{-1}C_{t,1}.
\]
Stacking up the above in a diagonal fashion for all nodes $t$ in
layer 2 will give exactly (\ref{eq:temp11}). Hence, our claim of
$A_{12}^{T}\tilde{A}_{11}^{-1}A_{12}\le\eta\tilde{A}_{22}$ holds.

Substituting the claim above into (\ref{eq:temp10}) yields
\begin{align}
 & \Delta x_{1}(d_{1})^{T}A_{11}\Delta x_{1}(d_{1})\nonumber \\
\le & \eta(\check{x}_{n}^{ML})^{T}(A_{(n-1)n}^{T}\tilde{A}_{(n-1)(n-1)}^{-1})\ldots(A_{23}^{T}\tilde{A}_{22}^{-1})\nonumber\\
&\cdot \tilde{A}_{22}(\tilde{A}_{22}^{-1}A_{23})\ldots(\tilde{A}_{(n-1)(n-1)}^{-1}A_{(n-1)n})\check{x}_{n}^{ML}.\label{eq:temp12}
\end{align}
Similarly, it can also be shown that $A_{23}^{T}\tilde{A}_{22}^{-1}A_{23}\le\eta\tilde{A}_{33}$
and $\check{C}_{3,2}^{T}\check{R}_{3,2}^{-1}\check{C}_{3,2}\ge A_{23}^{T}\tilde{A}_{22}^{-1}A_{23}$.
Therefore, the process leading to (\ref{eq:temp12}) can repeat, until
we get
\begin{align*}
 & \Delta x_{1}(d_{1})^{T}A_{11}\Delta x_{1}(d_{1})\\
\le & \eta^{d_{1}}(\check{x}_{n}^{ML})^{T}A_{(n-1)n}^{T}\tilde{A}_{(n-1)(n-1)}^{-1}A_{(n-1)n}\check{x}_{n}^{ML}\\
\le & \eta^{d_{1}}(\check{x}_{n}^{ML})^{T}\check{C}_{n,n-1}^{T}\check{R}_{n-1,n}^{-1}\check{C}_{n,n-1}\check{x}_{n}^{ML}\\
= & \eta^{d_{1}}\sum_{(t,j)}(x_{j}^{ML})^{T}C_{j,t}^{T}R_{t,j}^{-1}C_{j,t}x_{j}^{ML},
\end{align*}
where $(t,j)$ are such that node $t$ is $d_{1}$ hops away from
node 1 and node $j$ is connected to node $t$ but $d_{1}+1$ hops
away from node 1. Finally, it is easy to verify that
\begin{align*}
A_{11} & =C_{1}^{T}R_{1}^{-1}C_{1}+\sum_{j\in\mathcal{N}_{1}}C_{1,j}^{T}R_{1,j}^{-1}C_{1,j}=Q_{1}(1).
\end{align*}
Hence, (\ref{eq:acc_est}) holds.
\end{proof}

Similar to that in Corollary~\ref{cor:Q_inf_acc}, the accuracy of $\Delta x_{1}(\infty)$ is given in the Corollary~\ref{cor:accuracy-est}.
\begin{cor}
\label{cor:accuracy-est}Let\footnote{recall from Section~\ref{sec:ConvergenceX} that $S$ is an ordered
sequence of all $(i\rightarrow i,j)$}
\begin{align*}
B(k) & =Q^{1/2}(k+1)A(k)Q^{-1/2}(k)
\end{align*}
with $A(k)$ defined in \eqref{eq:dynamics} and
\begin{align*}
Q(k) & =\mathrm{diag}\left(Q_{i\rightarrow i,j}(k):(i\rightarrow i,j)\in S\right).
\end{align*}
Denote the maximum eigenvalue of $B(\infty)$ by $\beta=\overline{\mathrm{eig}}\left(B(\infty)\right)<1$. Under Assumption~\ref{asu:1}, if $d_{1}$
is large enough so that $\rho^{d_{1}-1}\simeq0$ and $\beta^{d_{1}-1}\simeq0$,
then
\[
\left\Vert \Delta x_{1}(\infty)\right\Vert ^{2}\lesssim\kappa\eta^{d_{1}}\left\Vert Q_{1}^{-1}(1)\right\Vert .
\]
\end{cor}
\begin{proof}
Let $\alpha(k)$ be the column vector formed by stacking up all the
$\alpha_{i\rightarrow i,j}(k)$ according to the index set $S$. From~\eqref{eq:BP-Q},\eqref{eq:BP-x-hat},\eqref{eq:BP-Q2}
and~\eqref{eq:BP-R}, it is straightforward to obtain
\[
\hat{x}_{1}(k)=f_{1}(k)-F_{1}(k)\alpha(k-1),
\]
with
\begin{align*}
f_{i}(k) & =Q_{i}^{-1}(k)\\
&~~\cdot\left(C_{i}^{T}R_{i}^{-1}z_{i}+\sum_{j\in\mathcal{N}_{i}}C_{i,j}^{T}R_{i,j\rightarrow i}^{-1}(k-1)z_{i,j}\right),\\
F_{i}(k) & =\mathrm{row}\left[F_{i,\left(j\rightarrow i,j\right)}(k):j\in\mathcal{N}_{i}\right],
\end{align*}
and
\begin{align*}
F_{i,\left(j\rightarrow i,j\right)}(k)=&Q_{i}^{-1}(k)C_{i,j}^{T}R_{i,j\rightarrow i}^{-1}(k-1)\\
\cdot & C_{j,i}Q_{j\rightarrow i,j}^{-1}(k-1).
\end{align*}
It then follows that
\[
\delta \hat{x}_{1}(k)=\delta f_{1}(k)-F_{1}(k)\delta\alpha(k-1)-\delta F_1(k)\alpha(k-2),
\]
where $\delta\xi(k)=\xi(k)-\xi(k-1)$ for any vector or matrix sequence
$\xi(k)$, $k\in\mathbb{N}$. Now, since $\rho^{d_{1}-1}\simeq0$,
in view of Lemma~\ref{lem:03}, we can approximate all the matrices
$Q_{i\rightarrow i,j}(k)$, $Q_{i,j\rightarrow j}(k)$, $Q_{i}(k)$
and $R_{i,j\rightarrow j}(k)$, for all $i$ and $(i,j\rightarrow j)$,
by their respective asymptotic values with $k\geq d_i$. Doing so leads to
\begin{equation}
\delta \hat{x}_{1}(k)\simeq-F_{1}(\infty)\delta\alpha(k-1),\label{eq:aux-1}
\end{equation}
and therefore
\begin{equation}
\left\Vert \delta \hat{x}_{1}(k)\right\Vert \lesssim\left\Vert F_{1}(\infty)\right\Vert \left\Vert \delta\alpha(k-1)\right\Vert .\label{eq:aux-2}
\end{equation}

Recall the definition of $\beta_{i\rightarrow i,j}(k)$ in \eqref{betaXXX}, $\alpha(k)$ and $\beta(k)$ are the column vectors formed by stacking up all the
$\alpha_{i\rightarrow i,j}(k)$ and $\beta_{i\rightarrow i,j}(k)$ respectively, according to the index set $S$. From~(\ref{eq:dynamics}),
we obtain
\[
\alpha(k+1)=B(k)\alpha(k)+\beta(k).
\]
Hence
\begin{align*}
\delta\alpha(k)=&B(k-1)\delta\alpha(k-1)+\delta B(k-1)\alpha(k-2)\\
&+\delta\beta(k-1),
\end{align*}
and since $\rho^{d_{1}-1}\simeq0$,
\[
\delta\alpha(k)\simeq B(\infty)\delta\alpha(k-1),\text{ for all }k\geq d_{1}.
\]

Also, $\beta^{d_{1}-1}\simeq0$ implies that
\begin{equation}
\left\Vert \delta\alpha(k)\right\Vert \simeq0,\text{ for all }k\geq d_{1}.\label{eq:aux-3}
\end{equation}

From~(\ref{eq:acc_est}) we get
\[
\left\Vert \Delta x_{1}(d_{1})\right\Vert ^{2}\leq\kappa\left\Vert Q_{1}^{-1}(1)\right\Vert\eta^{d_1},
\]
Hence, from~(\ref{eq:aux-2}) and~(\ref{eq:aux-3}),
\begin{align*}
 & \left\Vert \Delta x_{1}(\infty)\right\Vert \\
\leq & \left\Vert \Delta x_{1}(d_{1})\right\Vert +\sum_{k=d_{1}+1}^{\infty}\left\Vert \delta \hat{x}_{1}(k)\right\Vert \\
\simeq & \sqrt{\kappa\eta^{d_{1}}\left\Vert Q_{1}^{-1}(1)\right\Vert }+\left\Vert F_{1}(\infty)\right\Vert \sum_{k=d_{1}+1}^{\infty}\left\Vert \delta\alpha(k-1)\right\Vert \\
\simeq & \sqrt{\kappa\eta^{d_{1}}\left\Vert Q_{1}^{-1}(1)\right\Vert }.
\end{align*}
\end{proof}

\begin{rem}
\label{rem:acc_est} The results in Theorem~\ref{thm:accuracy_estimate} and Corollary~\ref{cor:accuracy-est}
show in a very quantitative way that the accuracy of the state estimate depends explicitly on 1) The number
of links connecting $\mathcal{G}_{i}(d_{i})$ and outside; 2) the
``size'' of the state for each such node as measured by $(x_{j}^{ML})^{T}C_{j,t}^{T}R_{t,j}^{-1}C_{j,t}x_{j}^{ML}\approx x_{j}^{T}C_{j,t}^{T}R_{t,j}^{-1}C_{j,t}x_{j}$;
3) the decay rate $\eta$; 4) cycle-free depth $d_{i}$. Accurate
state estimates require a combination of fast decay rate, large cycle-free
depth, small number of links connecting the inside and outside of
the cycle-free region, and small state "sizes" for such nodes.
\end{rem}

% \begin{remrk}\label{rem:accu_est1}%Note that Theorem~\ref{thm:accuracy_estimate} determines the accuracy of $\hat{x}_i(d_i)$ only. It is often important also to know the accuracy of $\hat{x}_i(\infty)$. This needs to be determined in combination of the convergence analysis of $\hat{x}_i$ in Theorem~\ref{thm:estimate.1} or Theorem~\ref{thm:estimate2}.{\bf \textcolor{red}{Should we add the expression of accuracy of $\hat{x}_i(\infty)$? But I think it is not in a explicit form.}}% \end{remrk}

\section{Simulations}

\label{sec:simulation}

In this section, we provide simulations to show that our convergence conditions together with its exponential decay curve.%two results:
%\begin{itemize}
%\item Different canonical graph topologies have different convergence properties
%for the state estimates, even though each node has the same
%measurement noises. In particular, the canonical graph with a smaller
%average number of neighbors and \textcolor{blue}{larger cycle-free depths} tends to converge
%more easily. Furthermore, the convergence is upper bounded by a exponential
%decay curve.
%\item The convergence of the information matrix is also
%upper bounded by an exponential decay curve.
%\item With the same measurement noises for each node, suppose the Gaussian
%BP estimates converge. Then, nodes with longer loop-free depths tend
%to have better accuracies.
%\item With the same measurement noises for each node, the information matrices
%of the nodes with longer loop-free depths tend to be closer to the
%inverse of estimation error covariance from WLS.
%\end{itemize}

The two examples are represented by the two graphs in Fig.~\ref{twograph}.
It is clear that each node in the right graph has more neighbors
and less {\em cycle-free depth}, when compared with the left graph.
\begin{figure}[ht]
\begin{centering}
\includegraphics[width=8cm]{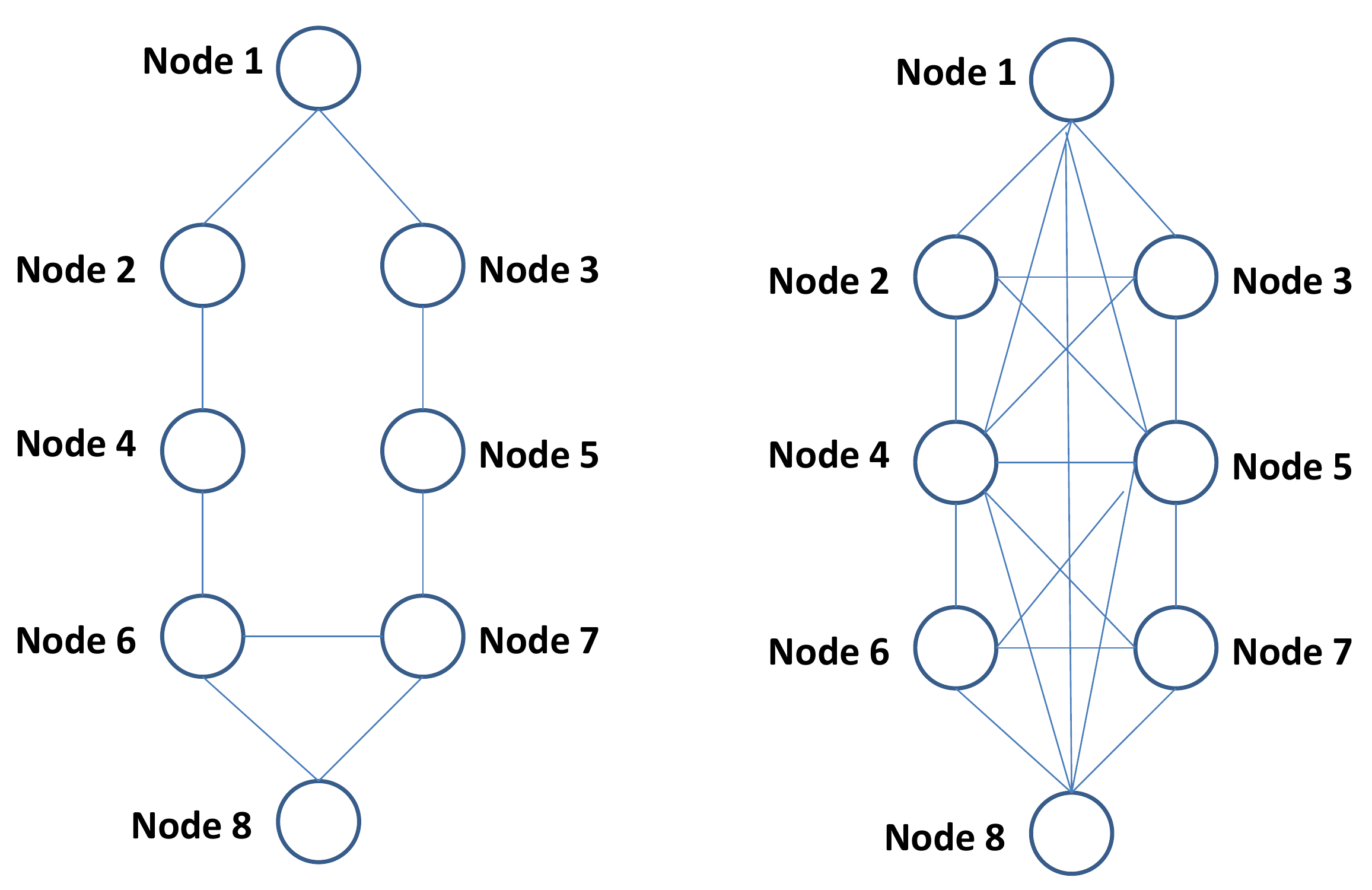}
\par\end{centering}

\caption{Comparison of two graphs with different connectivities}
\label{twograph}
\end{figure}

We assume that each node $i$ has self measurement and joint measurement
as follows:
\begin{align*}
z_{i} & =x_{i}+v_{i};\\
z_{i,j} & =x_{i}+x_{j}+v_{i,j}.
\end{align*}
with $R_{i}=5,R_{i,j}=1$, for all $j\in\mathcal{N}_{i}$.

%Firstly, the estimates of Node~1 in both graphs are compared to demonstrate the effect to convergence of the estimate, and the result is shown in~Fig.~\ref{BPconverge}. Since both the maximum singular values of $\bar{A}_6(\infty)$ and $\bar{A}_7(\infty)$ for the left graph are less than one, while the the maximum singular value of $\bar{A}_1(\infty)$ for the right graph is larger than one. We see that the Gaussian BP estimate for the graph on the left converges while that for the graph on the right does not. Furthermore, the theoretic upper bound of decay is depicted to verify our claim in the Theorem~\ref{thm:estimate.1}.

In the first plot we compare the evolution of the
estimate $\hat{x}_{1}(k)$, at node~1, in both graphs. We have that,
for the left graph, only nodes 6 and 7 have three or more neighbors,
and for these nodes, the maximum singular values of $\bar{A}_{6}(\infty)$
and $\bar{A}_{7}(\infty)$ are less than one. On the other hand, for
the right graph, the maximum singular value of $\bar{A}_{1}(\infty)$
is larger than one. Hence, according to Theorem~\ref{thm:estimate2},
the state estimate for the graph on the left should converge
while that for the graph on the right should not. These claims are
confirmed in Fig.~\ref{BPconverge}, whose $y$-axis is $\|\hat{x}_{1}(k)\|$ and $x$-axis is $k$. %shows the evolution of
%}\textbf{\textcolor{blue}{XXX what are you plotting, $||\hat{x}_{1}(k)-\hat{x}_{1}^{ML}(k)||$?
%Write here the formula, and add it also on the y-axis of the figure.
%Add also $k$ to the x axis.}}\textcolor{red}{I want to use the formula
%on y-axis, but I didn't find how to add formula in the ``.fig'' file?}
The figure also shows the theoretical decay rate upper bound
for $\|\hat{x}_{1}(k)\|$ as stated in Theorem~\ref{thm:estimate2}, i.e., the maximum eigenvalue
of $\bar{A}_{i}(\infty)\bar{A}_{i}^{T}(\infty), i=1,2,\ldots,I$.%\ref{thm:estimate.1}.
%}\textbf{\textcolor{blue}{XXX I think that we should refer to the
%result in Theorem~\ref{thm:estimate2} rather than~\ref{thm:estimate.1}.
%Also, Theorems~\ref{thm:estimate2} doesn't actually state any bound.
%It only says that convergence is exponential, and gives the exponent.
%Should we modify the statement of that theorem?}}
%\textcolor{red}{The
%theoretic upper bound is the $\bar{\rho}$(the maximum eigenvalue
%of $\bar{A}_{i}(\infty)\bar{A}_{i}^{T}(\infty)$) in Theorem~5, this
%figure shows that the actual decay rate is larger than the one in
%the sufficient condition(Theorem~5)}

\begin{figure}[ht]
\begin{centering}
\includegraphics[width=8cm]{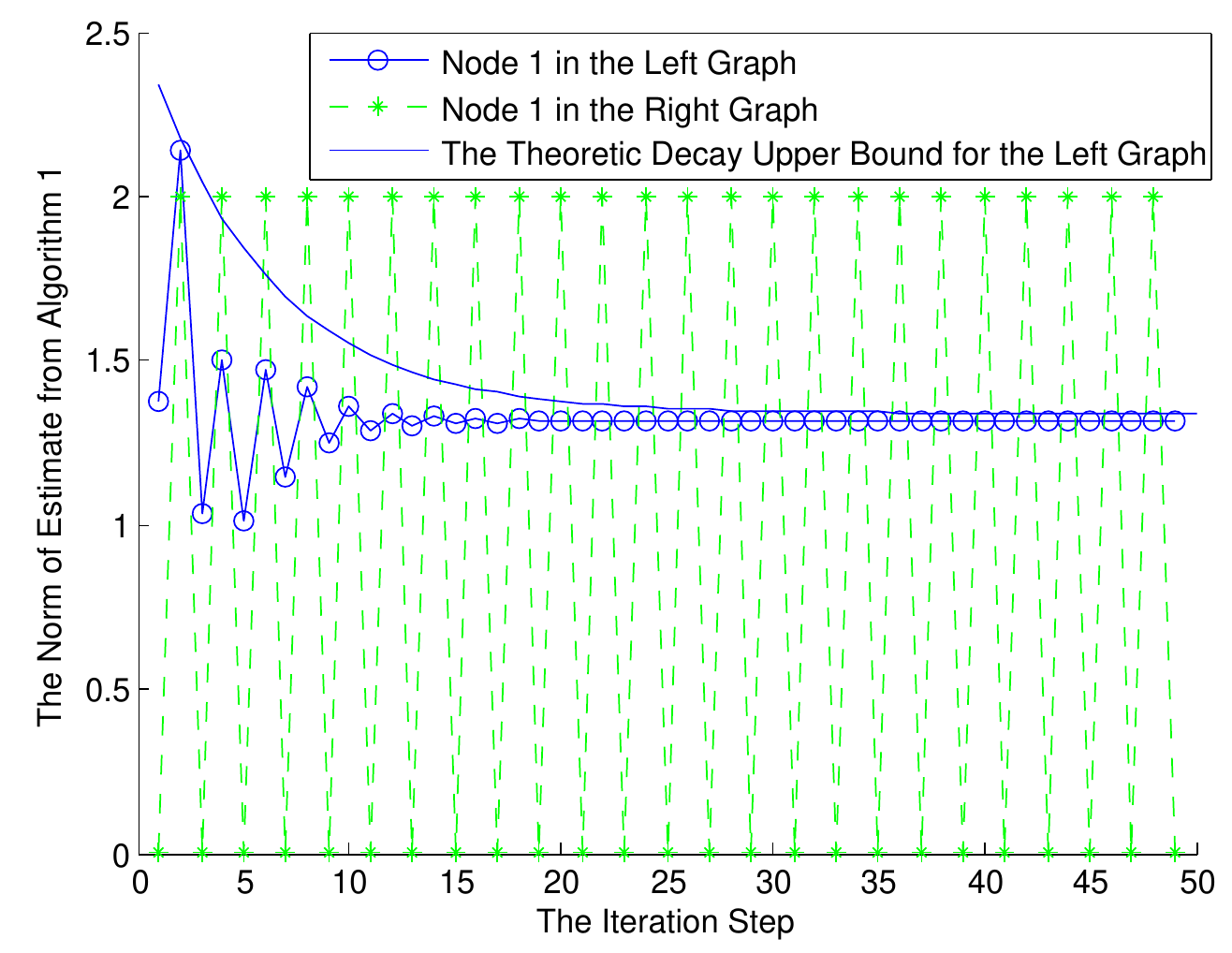}
\par\end{centering}

\caption{Convergence comparison of the state estimates for the two
graphs in Fig.~\ref{twograph}}
\label{BPconverge}
\end{figure}

%Secondly, since the information matrix by running BP Algorithm always converges, we choose Node~1 in the left graph of Fig.~\ref{twograph} as an example. The information matrix of Node~1 and its theoretic upper bound are depicted in the Fig.~\ref{BPconvergeQ}, which confirms the effectiveness of Theorem~\ref{thm:01}.

In the second plot we show the convergence of the
information matrix $Q_1(k)$ on the left graph of Fig.~\ref{twograph} and its corresponding upper bound for decay rate .
We know from Theorem~\ref{thm:info} that the information
matrix always exponentially converges with a upper bounded decay rate. This is illustrated in Fig.~\ref{BPconvergeQ}, whose $y$-axis is Trace$(Q_1(k))$ and $x$-axis is $k$.

\begin{figure}[ht]
\begin{centering}
\includegraphics[width=8cm]{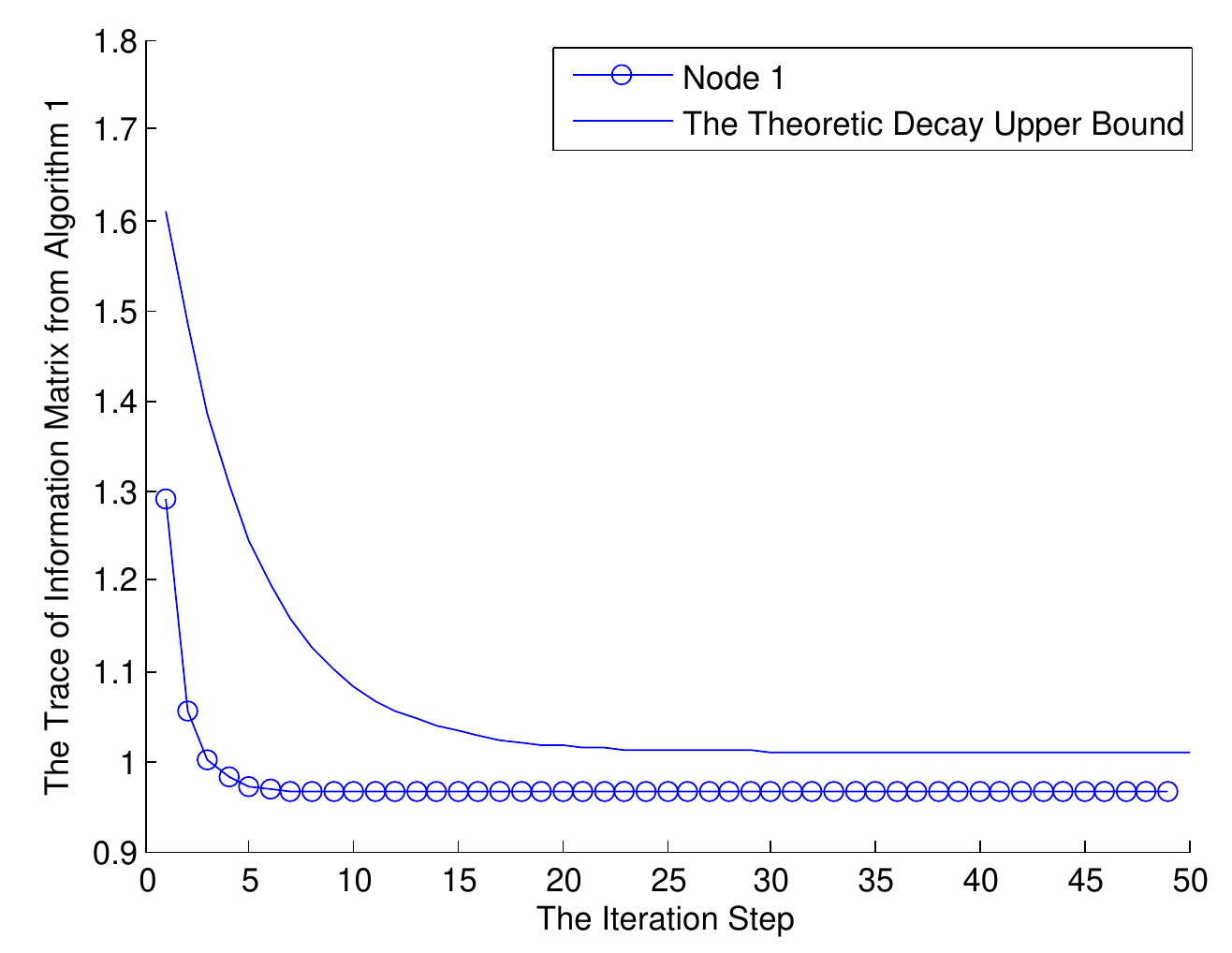}
\par\end{centering}
\caption{Convergence of the information matrices for
the left graph in Fig.~\ref{twograph}}
\label{BPconvergeQ}
\end{figure}
It is noted that, for a static estimation problem, the centralized optimal state estimate is deterministic and it follows that the local estimate converges to a constant vector.

\section{Conclusion}\label{sec:conclusion}
In this paper, a new distributed estimator for static systems is proposed in Algorithm~\ref{BP}. By viewing its iterations as a dynamic
process, we have carried out a complete analysis for its convergence and accuracy.
We have given conditions under which the Algorithm~\ref{BP} is guaranteed to converge, and we have provided concrete characterizations
of its accuracy. The influence of the so-called
{\em cycle-free depth} of each node to the accuracy is exploited. As explained in the Remark~\ref{rem:dynamic}, the Algorithm~\ref{BP} can also be effective in the dynamic state estimation by running several times during the time update of system sampling. Our results are expected to yield a
theoretical understanding of the distributed state estimation and may generate more applications for
this powerful algorithm.

\appendix

\subsection{Proof of Lemma~\ref{lem:01}}\label{app:lem:01}
We proceed by induction. To this end, we also argument the initialization
step of Algorithm 1 by setting $Q_{i\rightarrow i,j}(0)=\infty$.
By (\ref{eq:BP-Q2})-(\ref{eq:BP-R}), this yields $R_{i,j\rightarrow j}(0)=R_{i,j}$
and $Q_{i,j\rightarrow j}(0)=C_{j,i}^{T}R_{i,j}^{-1}C_{j,i}$. Then,
using (\ref{eq:BP-Q}) and (\ref{eq:BP-Q1}), we get
\begin{eqnarray*}
Q_{i\rightarrow i,j}(1)&=&C_{i}^{T}R_{i}^{-1}C_{i}+\hspace{-1mm}\sum_{w\in\mathcal{N}_{i}\backslash j}\hspace{-1mm}C_{i,w}^{T}R_{w,i}^{-1}C_{i,w}\\
&<&Q_{i\rightarrow i,j}(0).
\end{eqnarray*}
It follows again from (\ref{eq:BP-Q2})-(\ref{eq:BP-R}) that $R_{i,j\rightarrow j}(1)\ge R_{i,j\rightarrow j}(0)$
and $Q_{i,j\rightarrow j}(1)\le Q_{i,j\rightarrow j}(0)$. Hence,
the monotonicity properties (\ref{mono}) hold for $k=0$.

Now suppose the monotonicity properties (\ref{mono}) hold for some
$k\ge0$, we need to prove that they also hold for $k+1$. Indeed,
using (\ref{eq:BP-Q}) and (\ref{eq:BP-Q1}), we have
\begin{align*}
Q_{i\rightarrow i,j}(k+1) & =C_{i}^{T}R_{i}^{-1}C_{i}+\sum_{w\in\mathcal{N}_{i}\backslash j}Q_{i,w\rightarrow i}(k)\\
 & \le C_{i}^{T}R_{i}^{-1}C_{i}+\sum_{w\in\mathcal{N}_{i}\backslash j}Q_{i,w\rightarrow i}(k-1)\\
 & =Q_{i\rightarrow i,j}(k)
\end{align*}
Then, using (\ref{eq:BP-Q2})-(\ref{eq:BP-R}), we get $R_{i,j\rightarrow j}(k+1)\ge R_{i,j\rightarrow j}(k)$
and $Q_{i,j\rightarrow j}(k+1)\le Q_{i,j\rightarrow j}(k)$. By induction,
the monotonicity properties (\ref{mono}) hold for all $k\in\mathbb{N}$.
Finally, it is easy to verify that $\Omega_{i,j}=Q_{i\rightarrow i,j}(1)$
which implies $Q_{i\rightarrow i,j}(k)\le\Omega_{i,j}$ for $k\ge1$.

\subsection{Proof of Lemma~\ref{lem:02}}\label{app:lem:02}
We first prove that there exists a constant $\gamma>1$ such that
$R_{i,j\rightarrow j}(k)\le\gamma R_{i,j}$ for all $k\in\mathbb{N}$,
$1\le i\le I$ and $j\in\mathcal{N}_{i}$. We choose
\[
\gamma=\max\{(1-\eta)^{-1},\ \max_{i,j}\Vert R_{i,j}^{-1/2}R_{i,j\rightarrow j}(1)R_{i,j}^{-1/2}\Vert\}
\]
and proceed by induction. Since $\gamma I\geq R_{i,j}^{-1/2}R_{i,j\rightarrow j}(1)R_{i,j}^{-1/2}$,
we have $R_{i,j\rightarrow j}(1)\le\gamma R_{i,j}$. Now
suppose, for some $k\ge1$, we have $R_{i,j\rightarrow j}(k)\le\gamma R_{i,j}$
for all $i,j$. Then, for any $i,j$, we have $Q_{i,j\rightarrow j}(k)\ge\gamma^{-1}C_{j,i}^{T}R_{i,j}^{-1}C_{j,i}$
and
\begin{align*}
Q_{i\rightarrow i,j}(k+1) & \ge C_{i}^{T}R_{i}^{-1}C_{i}+\gamma^{-1}\sum_{w\in\mathcal{N}_{i}\backslash j}C_{i,w}^{T}R_{i,w}^{-1}C_{i,w}\\
 & \ge\gamma^{-1}\Omega_{i,j}.
\end{align*}
It follows that
\[
R_{i,j\rightarrow j}(k+1)\le R_{i,j}+\gamma C_{i,j}\Omega_{i,j}^{-1}C_{i,j}^{T}.
\]
Hence, it suffices to show that
\[
R_{i,j}+\gamma C_{i,j}\Omega_{i,j}^{-1}C_{i,j}^{T}\le\gamma R_{i,j},
\]
or equivalently,
\[
R_{i,j}\ge(\gamma-1)^{-1}\gamma C_{i,j}\Omega_{i,j}^{-1}C_{i,j}^{T}.
\]
Using Schur's complement, the above is the same as showing
\[
(\gamma-1)\gamma^{-1}\Omega_{i,j}\ge C_{i,j}^{T}R_{i,j}^{-1}C_{i,j}.
\]
Recalling the definition of $\gamma$, we have $(\gamma-1)\gamma^{-1}\ge\eta$.
From (\ref{eq:as1}), we get
\[
(\gamma-1)\gamma^{-1}\Omega_{i,j}\ge\eta\Omega_{i,j}\ge C_{i,j}^{T}R_{i,j}^{-1}C_{i,j}.
\]
By induction, $R_{i,j\rightarrow j}(k+1)\le\gamma R_{i,j}$ and $Q_{i\rightarrow i,j}(k+1)\ge\gamma^{-1}\Omega_{i,j}$
for all $k\in\mathbb{N}$. Combining with the monotonicity property
in Lemma~\ref{lem:01}, it completes the proof.

\subsection{Proof of Lemma~\ref{lem:03}}\label{app:lem:03}
From Lemma~\ref{lem:01},
we have
\begin{align*}
0 & \le\Delta Q_{i\rightarrow i,j}(k)\le\Delta Q_{i\rightarrow i,j}(k-1);\\
0 & \le\Delta Q_{i,j\rightarrow j}(k)\le\Delta Q_{i,j\rightarrow j}(k-1);\\
0 & \le-\Delta R_{i,j\rightarrow j}(k)\le-\Delta R_{i,j\rightarrow j}(k-1).
\end{align*}
Also define
\begin{align*}
\nabla Q_{i\rightarrow i,j}(k) & =Q_{i\rightarrow i,j}(k)-Q_{i\rightarrow i,j}(\infty)\\
\nabla Q_{i,j\rightarrow j}(k) & =Q_{i,j\rightarrow j}(k)-Q_{i,j\rightarrow j}(\infty)\\
\nabla R_{i,j\rightarrow j}(k) & =R_{i,j\rightarrow j}(k)-R_{i,j\rightarrow j}(\infty)
\end{align*}
and they have a similar monotonicity properties.

We proceed by induction. Recall $Q_{i\rightarrow i,j}(1)=\Omega_{i,j}$.
It follows that
\begin{eqnarray*}
\Delta Q_{i\rightarrow i,j}(1)&=&Q_{i\rightarrow i,j}^{-1/2}(\infty)Q_{i\rightarrow i,j}(1)Q_{i\rightarrow i,j}^{-1/2}(\infty)-I\\
&\le&\alpha I.
\end{eqnarray*}
So, \eqref{eq:01} holds for $k=1$.

Now suppose \eqref{eq:01} holds for some $k\ge1$. We need to show
that \eqref{eq:01} also holds for $k+1$. Indeed,
\begin{align}
& \nabla Q_{i\rightarrow i,j}(k+1)\nonumber\\
= & \sum_{w\in\mathcal{N}_{i}\backslash j}\nabla Q_{w,i\rightarrow i}(k)\nonumber \\
= & \sum_{w\in\mathcal{N}_{i}\backslash j}C_{i,w}^{T}[R_{w,i\rightarrow i}^{-1}(k)- R_{w,i\rightarrow i}^{-1}(\infty)]C_{i,w}.\label{eq:xxxxx}
\end{align}
Denote by $R_{\infty}=R_{w,i\rightarrow i}(\infty)$, following the matrix inversion lemma\footnote{The matrix inversion lemma states that $(A-CD^{-1}C^{T})^{-1}=A^{-1}+A^{-1}C(D-C^{T}A^{-1}C)^{-1}C^{T}A^{-1}$ when
$D>0$ and \\
$A-CD^{-1}C^{T}>0$.} and \eqref{eq:BP-R}, we have
\begin{align}
&R_{w,i\rightarrow i}^{-1}(k)\nonumber\\
= & (R_{\infty}-C_{w,i}(Q_{w\rightarrow w,i}^{-1}(\infty)-Q_{w\rightarrow w,i}^{-1}(k))C_{w,i}^{T})^{-1}\nonumber \\
= & R_{\infty}^{-1}+R_{\infty}^{-1}C_{w,i}\bar{Q}_{w\rightarrow w,i}^{-1}(k)C_{w,i}^{T}R_{\infty}^{-1},\label{eq:xxxx}
\end{align}
where
\begin{eqnarray*}
\bar{Q}_{w\rightarrow w,i}(k)&=&(Q_{w\rightarrow w,i}^{-1}(\infty)-Q_{w\rightarrow w,i}^{-1}(k))^{-1}\\
&&-C_{w,i}^{T}R_{\infty}^{-1}C_{w,i}.
\end{eqnarray*}
Next, we note that (\ref{eq:01}) implies
\[
Q_{w\rightarrow w,i}(k)\le(1+\alpha\rho^{k-1})Q_{w\rightarrow w,i}(\infty)
\]
which in turn implies
\begin{eqnarray}
Q_{w\rightarrow w,i}^{-1}(k)\ge\frac{1}{1+\alpha\rho^{k-1}}Q_{w\rightarrow w,i}^{-1}(\infty).\label{eq:asas}
\end{eqnarray}
%\begin{align*}
%& \Rightarrow Q_{w\rightarrow w,i}^{-1}(\infty)-Q_{w\rightarrow w,i}^{-1}(k)\\
%& \le(1-\frac{1}{1+\alpha\rho^{k-1}})Q_{w\rightarrow w,i}^{-1}(\infty)\\
%& =\frac{\alpha\rho^{k-1}}{1+\alpha\rho^{k-1}}Q_{w\rightarrow w,i}^{-1}(\infty),
%\end{align*}
%\[
%\Rightarrow (Q_{w\rightarrow w,i}^{-1}(\infty)-Q_{w\rightarrow w,i}^{-1}(k))^{-1}\ge\frac{1+\alpha\rho^{k-1}}{\alpha\rho^{k-1}}Q_{w\rightarrow w,i}(\infty),
%\]
Since
\begin{eqnarray*}
R_{\infty}&=&R_{w,i}+C_{w,i}Q_{w\rightarrow w,i}^{-1}(\infty)C_{w,i}^{T}\\
&>& C_{w,i}Q_{w\rightarrow w,i}^{-1}(\infty)C_{w,i}^{T}
\end{eqnarray*}
which, by Schur complement, implies
\[
Q_{w\rightarrow w,i}(\infty)>C_{w,i}^{T}R_{\infty}^{-1}C_{w,i}.
\]
Combining with \eqref{eq:asas}, we have
\begin{align*}
& (Q_{w\rightarrow w,i}^{-1}(\infty)-Q_{w\rightarrow w,i}^{-1}(k))^{-1}-C_{w,i}^{T}R_{\infty}^{-1}C_{w,i}\\
\ge & \frac{1}{\alpha\rho^{k-1}}Q_{w\rightarrow w,i}(\infty)+(Q_{w\rightarrow w,i}(\infty)-C_{w,i}^{T}R_{\infty}^{-1}C_{w,i})\\
\ge & \frac{1}{\alpha\rho^{k-1}}Q_{w\rightarrow w,i}(\infty).
\end{align*}

Substituting the above into (\ref{eq:xxxx}) and using (\ref{rho}), we get
\begin{align}
& R_{w,i\rightarrow i}^{-1}(k)-R_{\infty}^{-1}\nonumber\\
\le & \alpha\rho^{k-1}R_{\infty}^{-1}C_{w,i}Q_{w\rightarrow w,i}^{-1}(\infty)C_{w,i}^{T}R_{\infty}^{-1}\nonumber \\
\le & \alpha\rho^{k-1}R_{\infty}^{-1/2}[R_{\infty}^{-1/2}C_{w,i}Q_{w\rightarrow w,i}^{-1}(\infty)C_{w,i}^{T}R_{\infty}^{-1/2}]R_{\infty}^{-1/2}\nonumber \\
\le & \alpha\rho^{k}R_{\infty}^{-1}.\label{R-inverse}
\end{align}
Substituting the above into (\ref{eq:xxxxx}), we get
\begin{align*}
\nabla Q_{i\rightarrow i,j}(k+1) & \le\alpha\rho^{k}\sum_{w\in\mathcal{N}_{i}\backslash j}C_{i,w}^{T}R_{w,i\rightarrow i}^{-1}(\infty)C_{i,w}\\
 & \le\alpha\rho^{k}Q_{i\rightarrow i,j}(\infty).
\end{align*}
Hence,
\[
\Delta Q_{i\rightarrow i,j}(k+1)\le\alpha\rho^{k}I.
\]
By induction, (\ref{eq:01}) holds for all $k\in\mathbb{N}$.

\subsection{Proof of Lemma~\ref{lem:prop}}\label{app:lem:prop}
The diagonal elements of $A(\infty)$ being zero comes from the fact
that $\alpha_{i\rightarrow i,j}(k)$ does not feed into $\alpha_{i\rightarrow i,j}(k+1)$.
Next, we have
\begin{align*}
 & A_{i\rightarrow i,j}(\infty)A_{i\rightarrow i,j}^{T}(\infty)\\
= & Q_{i\rightarrow i,j}^{-1/2}(\infty)\{\sum_{w\in\mathcal{N}_{i}\backslash j}a_{w\rightarrow i,w}(\infty)a_{w\rightarrow i,w}^{T}(\infty)\}Q_{i\rightarrow i,j}^{-1/2}(\infty)\\
= & Q_{i\rightarrow i,j}^{-1/2}(\infty)\{\sum_{w\in\mathcal{N}_{i}\backslash j}C_{i,w}^{T}R_{i,w\rightarrow i}^{-1/2}(\infty)\Pi_{i,w}(\infty)\\
&\cdot R_{i,w\rightarrow i}^{-1/2}(\infty)C_{i,w}\}Q_{i\rightarrow i,j}^{-1/2}(\infty)
\end{align*}
where
\begin{align*}
\Pi_{i,w} & =R_{i,w\rightarrow i}^{-1/2}(\infty)C_{w,i}Q_{w\rightarrow i,w}^{-1}(\infty)C_{w,i}^{T}R_{i,w\rightarrow i}^{-1/2}(\infty)\\
&\le\rho I
\end{align*}
follows from (\ref{rho}). Using the above, combining with
(\ref{eq:BP-Q}), \eqref{eq:BP-Q1} and (\ref{eq:BP-Q2}), we get
\begin{align*}
 & A_{i\rightarrow i,j}(\infty)A_{i\rightarrow i,j}^{T}(\infty)\\
\le & \rho Q_{i\rightarrow i,j}^{-1/2}(\infty)\{\sum_{w\in\mathcal{N}_{i}\backslash j}C_{i,w}^{T}R_{i,w\rightarrow i}^{-1}(\infty)C_{i,w}\}Q_{i\rightarrow i,j}^{-1/2}(\infty)\\
\le & \rho Q_{i\rightarrow i,j}^{-1/2}(\infty)Q_{i\rightarrow i,j}(\infty)Q_{i\rightarrow i,j}^{-1/2}(\infty)\\
= & \rho I.
\end{align*}

\subsection{Proof of Lemma~\ref{lem:prop2}}\label{app:lem:prop2}
Suppose there is a leaf node $t$ (i.e., a node with only one neighbour)
in $\mathcal{G}$. Let $s$ be its connecting node in $\mathcal{G}$.
Choose the sequence $S$ such that $(t\rightarrow t,s)$ is the last
element in the sequence. Then, because $t$ is a leaf node, the last
row of $A(\infty)$, i.e., $A_{t\rightarrow t,s}(\infty)$, is a zero
row. It follows that $A(\infty)$ is stable if and only if the matrix,
obtained by removing the last row and column of $A(\infty)$, is stable.
This is the same as removing the leaf node $t$ from $\mathcal{G}$.

The process above can be repeated until that all the leaf nodes are
removed and the remaining graph has no more leaf nodes or it is a
singleton. We then obtain the resulting $\bar{\mathcal{G}}$ and $\bar{A}(\infty)$.
Hence, $A(\infty)$ is stable if and only if $\bar{A}(\infty)$ is
stable. For the special case where $\mathcal{G}$ has no loops, it
is also clear from the argument above that $A(\infty)$ is always
stable.

\subsection{Proof of Lemma~\ref{lem:tree2_BP}}\label{app:lem:tree2_BP}
We first claim that $\hat{x}_{1}(d_{1})$ by running Algorithm~\ref{BP}
on the original graph $\mathcal{G}$ is the same as the maximum likelihood
estimate of $\check{x}_{1}$ for line graph $\mathcal{G}^{l}$ (i.e.,
Fig.~\ref{tree2}(b)) after removing the super node $n$ and forcing
the state $\check{x}_{n}=0$. Indeed, the removal of the super node
$n$ does not affect the estimate $\hat{x}_{1}(d_{1})$ because information
from this node will take $d_{1}+1$ iterations to arrive node
1. Also, setting $\check{x}_{n}=0$ is the same as initializing $\check{\alpha}_{n-1,n}(0)=\check{C}_{n-1,n}^{T}\check{R}_{n-1,n}^{-1}\check{z}_{n-1,n}$,
which is done in the Algorithm~\ref{BP}; see (\ref{newinitial}).
Hence, our claim holds. We note that setting $\check{x}_{n}=0$ has
the same effect as including $\check{z}_{n-1,n}$ as additional self
measurements for $\check{x}_{n-1}$. Since the trimmed graph $\mathcal{G}^{l}$(i.e.,
Fig.~\ref{tree2}(b)) is a line graph, its state estimate from Algorithm~\ref{BP} agrees with
the maximum likelihood estimate. Applying Lemma~\ref{lem:tree2}
on this graph, we obtain the equation (\ref{eq:tree2_BP}) and the
conclusion that $\hat{x}_{1}(d_{1})$ is the first block of the solution
$\check{\mathbf{x}}$.

\subsection{Proof of Lemma~\ref{lem:tree2_sol}}\label{app:lem:tree2_sol}
Let $\check{\mathbf{x}}^{ML}$ be obtained from $\check{x}^{ML}$
by removing the last row block. Then, from Lemma~\ref{lem:tree2},
we have
\[
\mathbf{A}\check{\mathbf{x}}^{ML}+EA_{(n-1)n}\check{x}_{n}^{ML}=\mathbf{B},
\]
where $E$ is the column block matrix with all entries zero, except
that the last entry is an identity. Combining this with Lemma~\ref{lem:tree2_BP}
and defining $\Delta\check{\mathbf{x}}=\check{\mathbf{x}}-\check{\mathbf{x}}^{ML}$,
we have
\[
\mathbf{A}\Delta\check{\mathbf{x}}=-EA_{(n-1)n}\check{x}_{n}^{ML}.
\]
Following the inverse formula for band matrices in the Theorem~3.1
of \cite{meurant1992review}, the solution to $\Delta\check{\mathbf{x}}$
is that its $i$-th block component equals to
\[
\Delta\check{x}_{i}=(-\tilde{A}_{ii}^{-1}A_{i(i+1)})\ldots(-\tilde{A}_{(n-1)(n-1)}^{-1}A_{(n-1)n})\check{x}_{n}^{ML}.
\]
The result (\ref{eq:AB_sol}) for $\Delta x_{1}(d_{1})=\Delta\check{x}_{1}$
thus follows.

\bibliographystyle{IEEEtran}
\bibliography{mybibf4}

\end{document}